%% file: main.tex
\documentclass[10pt]{article}
\usepackage{geometry}
\usepackage[utf8]{inputenc}
\usepackage{amsmath}
\usepackage{amssymb}
\usepackage{amsthm}
\usepackage{hyperref}
\usepackage{cleveref}
\usepackage{graphicx}
\usepackage{bm}
\usepackage{mathtools}
\usepackage{dsfont}
\usepackage{tikz}
\usepackage{tikz-cd}
\usepackage{booktabs}
\usepackage{multirow}
\usetikzlibrary{decorations.pathreplacing,calligraphy}
\usepackage[font=normalsize,labelfont=bf]{caption}
\usepackage{subcaption}

\newtheorem{theorem}{Theorem}
\newtheorem{proposition}{Proposition}

\newtheorem{lemma}[proposition]{Lemma}

\theoremstyle{remark}

\newcommand{\vocab}{\textit}
\renewcommand{\AA}{\mathcal{A}}
\newcommand{\CC}{\mathcal{C}}
\newcommand{\FF}{\mathcal{
C}}
\newcommand{\EE}{\mathbb{E}}
\newcommand{\RR
}{\mathbb{R}}
\newcommand{\DD}{\mathcal{D}}

\renewcommand{\epsilon}{\varepsilon}
\newcommand{\simiid}{\overset{\mathrm{iid}}{\sim}}
\newcommand{\eqdist}{\stackrel{d}{=}}

\DeclareMathOperator{\Var}{Var}

\usepackage{xcolor}
\newcommand{\Comments}{1}
\newcommand{\mynote}[3]{\ifnum\Comments=1\textcolor{#1}{#2: #3}\fi}

\usepackage[round,sort]{natbib}
\hypersetup{
    colorlinks,
    linkcolor={red!50!black},
    citecolor={blue!50!black},
    urlcolor={blue!80!black}
}

\newcommand\blfootnote[1]{%
  \begingroup
  \renewcommand\thefootnote{}\footnote{#1}%
  \addtocounter{footnote}{-1}%
  \endgroup
}

\title{Wisdom and Foolishness of Noisy Matching Markets}
\author{Kenny Peng\\ Nikhil Garg\\\{klp98,ngarg\}@cornell.edu\blfootnote{We thank Sophie Greenwood and Jon Kleinberg for helpful discussion and feedback.}}
\date{}

\begin{document}

\maketitle

\begin{abstract}
    We consider a many-to-one matching market where colleges share true preferences over students but make decisions using only independent noisy rankings. Each student has a \textit{true value} $v$, but each college $c$ ranks the student according to an independently drawn \textit{estimated value} $v + X_c$ for $X_c\sim \DD.$ We ask a basic question about the resulting stable matching: How noisy is the set of matched students? Two striking effects can occur in large markets (i.e., with a continuum of students and a large number of colleges). When $\DD$ is light-tailed, noise is fully attenuated: only the highest-value students are matched. When $\DD$ is long-tailed, noise is fully amplified: students are matched uniformly at random. These results hold for any distribution of student preferences over colleges, and extend to when only subsets of colleges agree on true student valuations instead of the entire market. More broadly, our framework provides a tractable approach to analyze implications of imperfect preference formation in large markets.
\end{abstract}

\input{introduction.tex}
\input{related-work.tex}
\input{model.tex}
\input{main-results.tex}
\input{proof-outline.tex}

\input{model-extended}

\input{discussion.tex}
\input{conclusion.tex}

{
\bibliography{bib}
}

\appendix

\input{proof-attenuating.tex}
\input{proof-amplifying.tex}
\input{proofs-extended.tex}

\end{document}

%% file: introduction.tex
\section{Introduction}
In two-sided matching---such as between firms and workers, hospitals and residents, or colleges and students---noise is inevitable. Firms evaluate job applicants using limited and imperfect information from resumes and interviews. Students, when deciding which school to attend, may have only anecdotal or coarse knowledge. Given this noise, do the ``correct'' matches still form? 

This article addresses this question, in a setting in which (subsets of) colleges share true preferences based on a measure we refer to as student quality. In a noise-free setting, the highest quality students match with their most preferred colleges. However, suppose each college makes offers based on independent noisy estimates of student quality. Do the highest-quality students still match? In essence, we are interested in how ``local'' noise introduced in the evaluation processes of individual colleges is aggregated by the market to produce the ``global'' outcome. Perhaps nothing remarkable happens---individual colleges make noisy decisions, and the set of matched students is ``just as noisy.'' But one can also imagine two types of market-level effects that are more interesting: an \textit{attenuating} effect, in which noise cancels out, or an \textit{amplifying} effect, in which noise compounds. We show that in large markets the most extreme versions of both effects can occur.

In the basic model we analyze, students each have a true value $v\in \RR$. Each college $c$ ranks students according to an estimated value $v + X_c$, where $X_c$ is drawn i.i.d. from some distribution $\DD$ (i.e., colleges form preference lists according to a random utility model). Given these noisy college preferences and \textit{any} distribution of student preferences, we analyze the probability that a student with true value $v$ is matched in the corresponding stable matching. When there is a continuum of students and the number of colleges grows large, we find that:
\begin{enumerate}
    \item If $\DD$ is light-tailed (i.e., when the maximum order statistic concentrates), the probability a student matches approaches a step function; all students with true value above a certain cutoff are guaranteed to match, and all students with value below the cutoff are guaranteed not to match, as in the noiseless case. Noise is fully attenuated.
    \item If $\DD$ is long-tailed (i.e., when the hazard rate vanishes), the probability a student matches approaches the same constant independent of their true value. Noise is fully amplified.
\end{enumerate}
These results---\Cref{thm:attenuating,thm:amplifying}, respectively---reveal that matching markets can both strongly attenuate and amplify noise, suggesting that either a ``wisdom of crowds'' or ``foolishness of crowds'' can occur in matching markets. (For noise distributions falling in between the two regimes above---such as the exponential and Gumbel distribution---it seems that a more moderate effect in either direction can occur; finding a precise characterization is an interesting open question.) 
Of course, at the most granular level, match outcomes will always remain noisy: a high-quality student may not receive an offer from a specific individual college; however, \Cref{thm:attenuating} says that there are regimes where the student is guaranteed to match to \textit{some} college.

\Cref{thm:attenuating,thm:amplifying} reflect and suggest an underlying insight: that whether a student matches depends on their \textit{highest} estimated values---for a student, what really matters is that at least one college views them highly. As an analogy, imagine that two students take tests. The two students have true skill levels $v$ and $v'$, but on a given test, the students receive scores $v + X$ and $v + X'$ where $X$ and $X'$ are i.i.d. drawn noise from some distribution $\DD$. The students take many tests. For someone who is trying to determine which student has a higher true skill level, the natural approach is to evaluate the students based on their average performance---an approach that will certainly work due to the law of large numbers (as long as $\DD$ has finite mean). The ``approach'' taken by a matching market, however, more closely resembles evaluating students based on their \textit{best} performance. This approach sometimes works very well (when the maximum of many samples from a distribution converges) but can also work very poorly (when the tail of the distribution is heavy), revealing our results' two regimes.
\begin{figure}
    \centering
    \begin{subfigure}{0.32\textwidth}
        \centering
        \includegraphics[width=\linewidth]{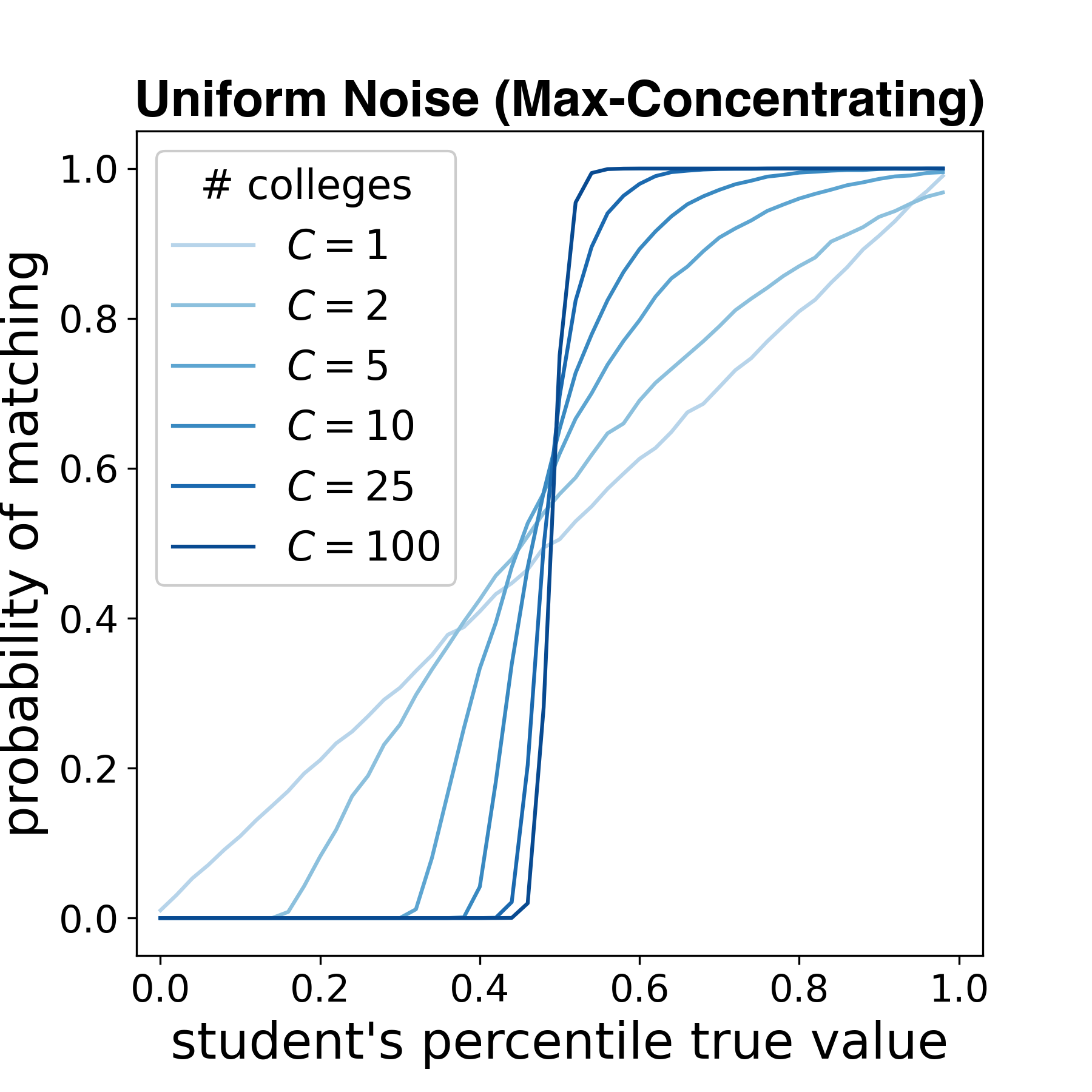}
        \caption{\centering Noise Attenuation as $C\rightarrow \infty$ (Wisdom of Crowds)}
        \label{fig:sub1}
    \end{subfigure}
    \begin{subfigure}{0.32\textwidth}
        \centering
        \includegraphics[width=\linewidth]{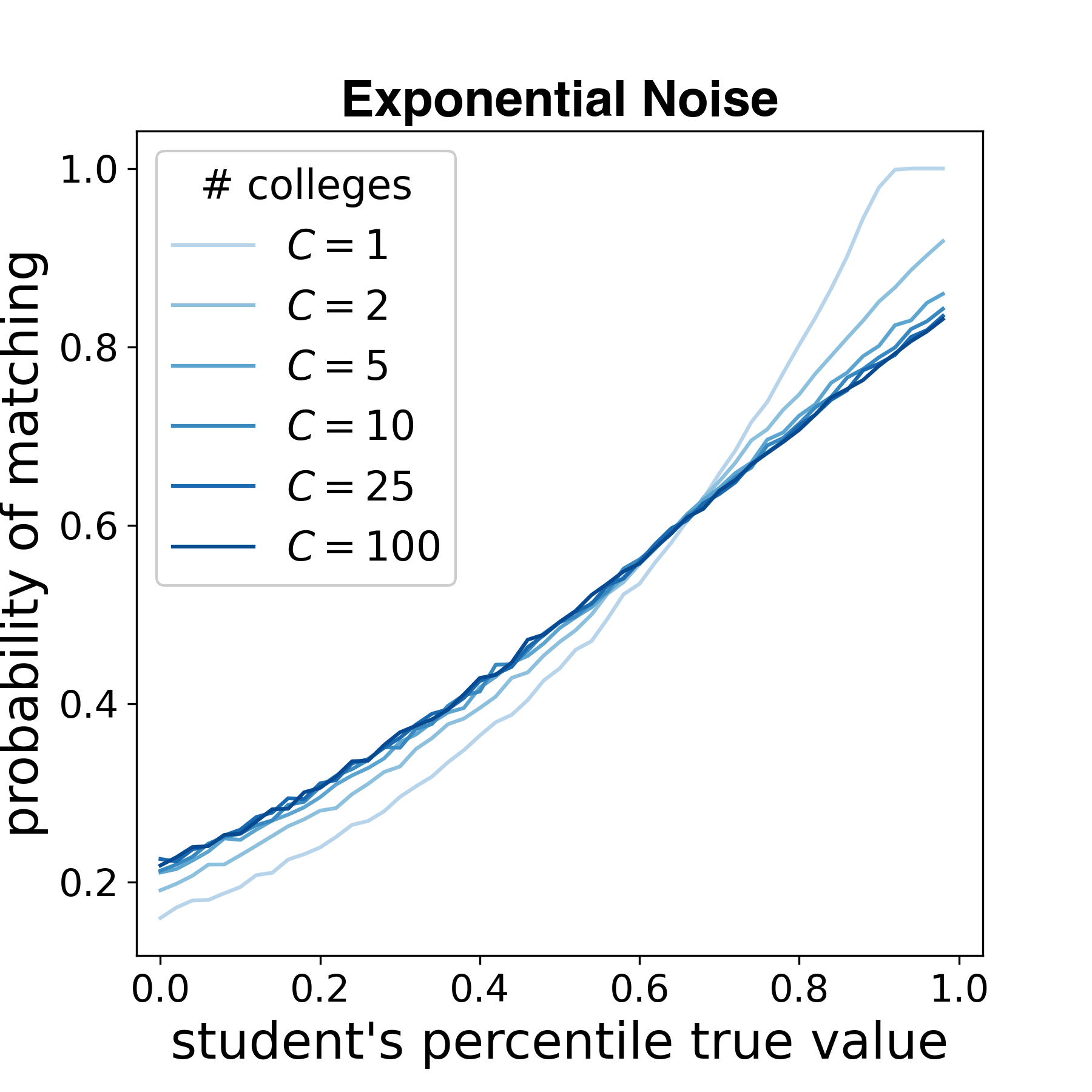}
        \caption{Noise Stays $\approx$ Same \\{\color{white}.}}
        \label{fig:sub2}
    \end{subfigure}
    \begin{subfigure}{0.32\textwidth}
        \centering
        \includegraphics[width=\linewidth]{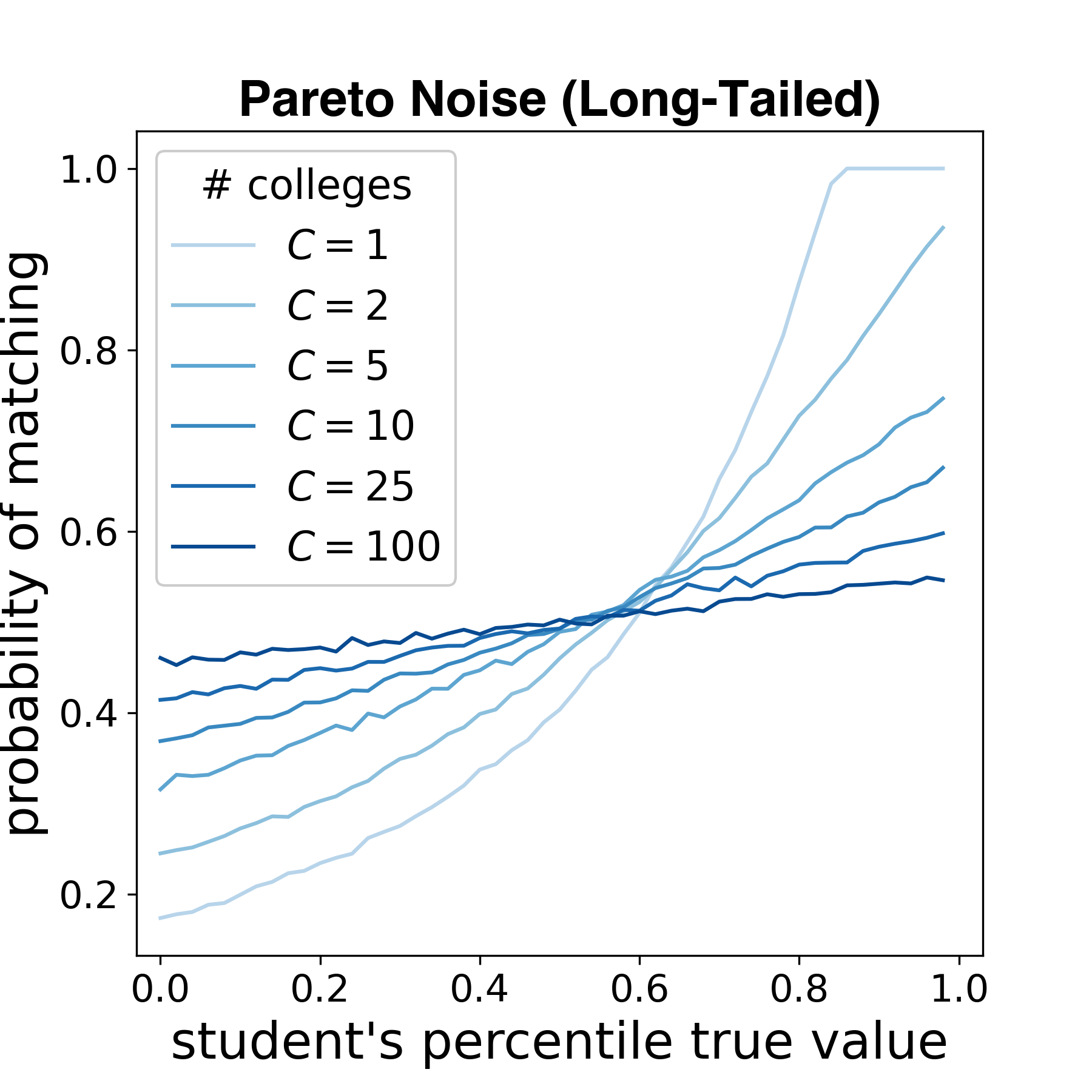}
        \caption{Noise Amplification as $C\rightarrow \infty$\\{\color{white}.} \hspace{0.5cm}  (Foolishness of Crowds)}
        \label{fig:sub3}
    \end{subfigure}
    \caption{We compare $p_\mu(v)$, the probability that a student with true value $v$ matches, across three economies that differ only in their noise distribution $\DD$. For uniform noise, which is max-concentrating, $p_\mu(v)$ approaches a step function as $C$ grows large (\Cref{thm:attenuating}). For Pareto noise, which is long-tailed, $p_\mu(v)$ approaches a constant as $C$ grows large (\Cref{thm:amplifying}). For exponential noise, which lies between our two regimes of focus, $p_\mu(v)$ does not change significantly as $C$ grows large. The economy here has $2000$ students and the total capacity of colleges is $1000$. Students have uniformly random preferences over colleges. Plots display averages over 100 simulations.}
    \label{fig:three_figures}
\end{figure}

Showing \Cref{thm:attenuating,thm:amplifying} requires a more subtle proof method, due to global interactions in matching markets. We use the ``cutoff structure'' of stable matching \citep{abdulkadirouglu2015expanding, azevedo2016supply}. In this characterization, each college $c$ has a cutoff $P_c$. Then a student can ``afford'' a college if and only if their ``score'' (here, their estimated value) at that college exceeds the cutoff; the student then matches with their most-preferred college among those they can afford. This matching is stable if and only if this process clears the market---i.e. when the capacities of colleges are properly met but not exceeded. In our setting, given market-clearing cutoffs $P_1, P_2, \cdots,$ a student with true value $v$ is matched if and only if there exists a college $c$ such that $v + X_c > P_c.$ Following the intuition of the previous paragraph, this roughly depends on the maximum $X_c$ a student obtains; however, it also depends on the value of the associated cutoff $P_c$, which depends on $\DD$, students preferences, and college capacities, complicating matters. Fortunately, we need not explicitly determine the market-clearing cutoffs from student and college preferences---a task that appears intractable; rather, the market-clearing conditions can be used indirectly.

While we first consider effects in a setup in which all colleges in an economy share the same true preferences over students, the results extend to a much broader setting. In \Cref{sec:extended}, we show that analogues of \Cref{thm:attenuating,thm:amplifying} hold for any large \textit{coalition} of colleges with the same true preferences, even when this coalition exists within a broader economy where students and other colleges have otherwise arbitrary preferences. Whether noise attenuates or amplifies within each coalition depends only on its noise distribution, allowing for arbitrary two-sided preferences outside the coalition.

Taking a step back, our approach suggests a more general recipe for analyzing the consequences of imperfect preference formation in markets: first, specify the true preferences of participants in a matching market; second, specify how participants form imperfect (e.g., incomplete, noisy, or biased) preferences in practice; third, compute the outcome of the market resulting from the imperfect preferences; finally, analyze the outcome with respect to participants' true preferences. We elaborate on this general framework in \Cref{sec:discussion}.

The article proceeds as follows. In \Cref{sec:related-work}, we discuss related work. In \Cref{sec:model}, we introduce the basic model. In \Cref{sec:results}, we introduce our main results, \Cref{thm:attenuating,thm:amplifying}. We give an overview of their proofs in \Cref{sec:proof-outline}, leaving full proofs to the appendix. In \Cref{sec:extended}, we give the extended model and results. \Cref{sec:discussion} discusses the general framework mentioned above, and \Cref{sec:conclusion} concludes.

%% file: related-work.tex
\section{Related Work}\label{sec:related-work}
Our work connects to several active lines of work in matching and school admissions \citep{gale1962college}, including stable matching under uncertainty, random and large matching markets, and uncertainty in admissions and hiring.

\paragraph{Matching with incomplete information.} A widely recognized limitation of stable matching theory is that agents are typically assumed to have sufficient information to determine their true preferences. Several strands of literature consider this limitation. One strand has developed stability concepts in the presence of incomplete information \citep{chakraborty2010two, liu2014stable, bikhchandani2017stability, liu2020stability, chen2020learning, chen2023theory}. This line of work focuses on the strategic behavior of agents under different mechanisms when agents have incomplete information about their preferences and may be further informed by the actions of other participants. 

A second strand of work focuses on (facilitating) the process of potentially costly information acquisition in stable matching \citep{gonczarowski2019stable, ashlagi2020clearing, immorlica2020information, kanoria2020facilitating, shi2023optimal}. In particular, this line of work has characterized the amount of information required to obtain stable matchings; \citet{gonczarowski2019stable} and \citet{ashlagi2020clearing} take explicitly information-theoretic approaches. More broadly, these works shed light on the extent to which information acquisition limitations prevent the formation of stable matchings, and suggest possible interventions to facilitate efficient information acquisition. Relatedly, several works focus on how agents can sequentially learn preferences and form matches
\citep{jagadeesan2023learning,liu2020competing, liu2021bandit, dai2021learning, ionescu2021incentives, jeloudar2021decentralized}. These works adopt a learning-theoretic approach, e.g., through the perspective of multi-armed bandits.

Our approach is distinct and complementary to this literature. The above works focus on settings in which agents (1) behave strategically depending on their incomplete information, or (2) (strategically) acquire additional information. Meanwhile, we assume agents report their preferences exactly according to the noisy signals they receive. We then analyze the structure of the resulting matching, obtaining results that are more ``statistical'' than game- or learning-theoretic. For example, we do not consider stability with respect to true preferences, or allow participants to learn their preferences over time; rather, we analyze the market-wide statistical interaction between noise and stable matching. While it is important to understand how participants behave strategically given incomplete information and how participants acquire additional information, in many settings it may be reasonable to assume that participants act ``naïvely'' given their best knowledge of their own preferences: for example, that schools rank their favorite candidates according to grades, geographic preferences, and perceptions of interviews, letters, and essays, without strategizing beyond the inherently limited set of information they have. 

\paragraph{Random and large matching markets.}
While sharing a similar motivation as the aforementioned work on incomplete information and information acquisition in matching markets, the ``statistical'' style of our work bears a greater resemblance to the analysis of random and large matching markets. The random markets literature includes a substantial line of results giving asymptotic bounds on the number of stable matchings and resulting welfare of agents in different random matching markets \citep{pittel1989average, immorlica2003marriage, ashlagi2017unbalanced, hoffman2023stable}. 

Our results and proof method rely heavily on a model of matching markets with a continuum of students. More generally, analysis of matching markets has benefited substantially from models of ``large markets '' \citep{kojima2017recent, leshno2023large}. The continuum model, in particular, induces an especially tractable ``cutoff structure'' of stable matching \citep{abdulkadirouglu2015expanding, azevedo2016supply}. The model we introduce is built directly on the model of \citet{azevedo2016supply}, extending it to a setting where colleges have noisy estimates of their preferences. This extension expands the scope of questions that can be studied using these models, as we discuss further in \Cref{sec:discussion}. 

\paragraph{Models with noisy estimates.} The basic structure of our model---in which students have a true value $v$, and where colleges only receive a noisy estimate $v + X$ of this value---can be found in many papers. For example, a number of papers present models of college admissions using similar frameworks \citep{chade2014student, emelianov2022fair,gargtesting21,liuoptionaltests21,niumultipletests22}. These papers focus on somewhat different questions: for example, \citet{castera2022statistical} study the effects of statistical discrimination in matching markets, \citet{chade2014student} analyze the portfolio problem of choosing which colleges to apply to, and \citet{gargtesting21} analyze the trade-off between accuracy and equity when using specific features (e.g., test scores) of students in college admissions. One innovation of our work is the adoption of large market models---especially, with a large number of colleges. Indeed, while the models above focus primarily on settings with one or two colleges, the most striking findings of our model emerge with a large number of colleges.

\citet{chade2006matching} gives a close search-theoretic analogue of our model in which agents on each side of the market sequentially receive noisy signals from agents on the other side, and must choose to accept or reject. \citet{chade2006matching} finds that assortative matching (i.e., the most desired agents are more likely match with each other) occurs when noise satisfies the monotone likelihood ratio property, and that the optimal strategy for agents requires accounting for the fact that acceptance from the other agent is ``bad news''---an analog of the ``winner's curse.'' While assortative matching is shown, the magnitude of sorting is not explicitly given: it may be interesting to consider if (as search costs grow small) analogs to our results can be found: that despite individually noisy signals, the resulting matching displays perfect sorting; or, if the noise distribution is different, that the matching displays no sorting.

\paragraph{Comparison to \citet{peng2023monoculture}.} Most related to this work is our prior work \citet{peng2023monoculture}, which studies algorithmic monoculture in matching markets.\footnote{\citet{kleinberg2021algorithmic} initiated the study of \textit{algorithmic monoculture}, the phenomenon in which many decision-makers rely on the same algorithm to evaluate applicants; consequently, they also study what happens when multiple decision-makers use independent but noisy evaluations. Indeed, the directional effect of our \Cref{thm:attenuating} is shown in \cite{kleinberg2021algorithmic} in the two-college case.} Both our prior and present work extend the model of \citet{azevedo2016supply}---with a continuum of applicants and discrete colleges---so that colleges form noisy preferences via a random utility model. In particular, our prior work \citep{peng2023monoculture} proves \Cref{thm:attenuating} of the present work in a fully symmetric setting where students have uniformly random preferences over colleges and colleges have identical capacities. \Cref{thm:attenuating} in the present work generalizes this result for arbitrary student preferences and college capacities, requiring a significantly more sophisticated proof method. 

The questions addressed by each work are substantively different. \citet{peng2023monoculture}, building on \citet{kleinberg2021algorithmic}, compares monoculture (when noise is fully correlated across colleges for each student) to polyculture (independent noise), in terms of who is matched, student welfare, and disparities when students submit different numbers of applications. In contrast, the present work focuses on the first question of who is matched under independent noise, but in a more general model with arbitrary student preferences. In addition to the general noise attenuation result of \Cref{thm:attenuating}, \Cref{thm:amplifying} of the present work gives a noise regime in which the \textit{opposite effect} of that found in \cite{kleinberg2021algorithmic} and \cite{peng2023monoculture} occurs, suggesting that the structure of noise plays an essential role. Moreover, \Cref{thm:attenuating-extended} and \Cref{thm:amplifying-extended} show that these results extend to large subsets of colleges that share true preferences. More broadly, both papers are instantiations of a framework we introduce to study imperfect preference formation in matching markets, expanded upon in \Cref{sec:discussion}.

%% file: model.tex
\section{Model}\label{sec:model}

We introduce a model of two-sided matching markets in which colleges have shared true preferences over students, but act according to independent noisy approximations of these preferences. Meanwhile, students can have arbitrary preferences over colleges. For tractability, we consider a continuum of students. Our model extends that of \citet{azevedo2016supply}, allowing for noisy colleges preferences. Indeed, we adopt much of the same notation. In \Cref{sec:extended}, we extend the basic model introduced in this section to relax the assumption that all colleges share true preferences.

To fully and rigorously specify an economy (and its corresponding stable matchings) will require a fairly significant number of parameters, summarized in \Cref{tab:params}. However, as perhaps foreshadowed by the generality of our results (informally stated in the introduction), many of these parameters will ``fade into the background'' by the time we reach our analysis.

\subsection{Definitions}
There is a continuum of students of measure $1$ and a finite set of colleges $C = \{1, 2, \cdots, C\}$. A standard model of stable matching between students and colleges has three ingredients: student preferences over colleges, college preferences over students, and college capacities. Given these three ingredients, we may analyze the corresponding set of \textit{stable matchings}, matchings between students and colleges so that no student-college pair would jointly prefer to defect from the matching.

The focal point of our model is on how colleges form preferences over students, so this is where we begin. Each student has a \vocab{value} $v\in \RR$, such that colleges all prefer students with higher value. Student values specify the true preferences of colleges. Colleges, however, only obtain noisy estimates of student values, dictated by a \textit{noise distribution} $\DD$: for a student with value $v$, each college $c\in C$ obtains a noisy estimate equal to $v + X_c$, where $X_c$ is drawn independently from $\DD$. Given this noise, we are interested in the probability that a student with value $v$ matches. For the sake of comparison, it is convenient to hold two parameters of our model fixed: 
\begin{itemize}
    \item $\eta$, a probability measure over student values, and
    \item $S\in (0,1)$, the total capacity of colleges.
\end{itemize}
We assume that $\eta$ has connected support, which implies that there is a unique value $v_S$ such that $\eta((v_S, \infty))=S$; i.e., the measure of students with value at least $v_S$ is equal to $S$, the total capacity of colleges. Holding $\eta$ and $S$ fixed is convenient, because we can specify two extreme behaviors. On one extreme, only students with value at least $v_S$ are matched, reflecting a ``noise free'' matching; on the other extreme, all students match with probability $S$, reflecting a ``fully noisy'' matching. Note that overall capacity is less than the measure of students.

We further assume a (weak) regularity condition on $\eta$: that there is a constant $\gamma>0$ such that $\eta((v, v+\delta)) = O(\delta^\gamma)$ for all $v\in \RR$ and $\delta > 0.$ This is equivalent to the cumulative distribution function of student values satisfying a $\gamma$-Hölder condition. We also fix a constant $\alpha>0$ that we use to specify a regularity condition on college capacities.
\\
\\
\begin{table}
\caption{Model parameters.}
\label{tab:params}
\begin{center}
\small
    \begin{tabular}{p{4cm} c p{7cm}}
    \toprule
        \multirow{6}{*}{Constants in our model} & $\eta$ & a probability measure over student values\\ \noalign{\vspace{0.5ex}}\cline{2-3}\noalign{\vspace{0.5ex}}
         & $S$ & the total capacity of colleges \\ \noalign{\vspace{0.5ex}}\cline{2-3}\noalign{\vspace{0.5ex}}
         & $\gamma$ & a constant governing a regularity condition on $\eta$\\
         \noalign{\vspace{0.5ex}}\cline{2-3}\noalign{\vspace{0.5ex}}
         & $\alpha$ & a constant governing a regularity condition on individual college capacities\\ \midrule
        \multirow{5}{*}{Parameters of an economy} & $\DD$ & a noise distribution\\ \noalign{\vspace{0.5ex}}\cline{2-3}\noalign{\vspace{0.5ex}}
         & $C$ & the total number of colleges\\ \noalign{\vspace{0.5ex}}\cline{2-3}\noalign{\vspace{0.5ex}}
         & $\vec{S}$ & a vector of college capacities\\ \noalign{\vspace{0.5ex}}\cline{2-3}\noalign{\vspace{0.5ex}}
         & $\rho$ & a probability measure over student types\\
    \bottomrule
    \end{tabular}
    \end{center}
\end{table}

With $\eta$ and $S$ fixed, we may fully parameterize an \vocab{economy} by the ordered tuple $(\DD, C, \vec{S}, \rho),$ where $\DD$ is the noise distribution, $C$ the total number of colleges, $\vec{S}=(S_1, \cdots, S_C)$ a vector of college capacities, and $\rho$ a probability measure over student types (to be defined). We require $\vec{S}$ and $\rho$ to satisfy some basic conditions:

As the total capacity $S$ is fixed, we must have that $S_1 + S_2 + \cdots + S_C = S.$ We also assume that 
\begin{equation}
    S_c < \frac{\alpha}{C}
\end{equation} 
for all $c\in C$ for some fixed $\alpha$, meaning that no college has a disproportionately large capacity. This kind of assumption is necessary for our results, since otherwise it is possible for an economy to technically have many colleges, but such that a few colleges take up most of the capacity; this results in the economy behaving as if it only had a few colleges. It will also be useful to define $S(C')$ as the total capacity of colleges in $C'\subseteq C$, that is, $\sum_{c\in C'} S_c.$ By definition, $S = S(C).$

To specify student preferences over colleges, we identify each student with their \vocab{type} $\theta = (v, \succ),$ where $v$ is their value and $\succ$ is their preference ordering over colleges. Then $\rho$ is a probability measure over $\Theta = \RR\times \mathcal{R},$ where $\mathcal{R}$ is the set of strict preference orderings over colleges. In other words, $\rho$ induces a joint distribution over student values and student preferences. We assume that all students prefer all colleges over being unmatched. We allow the distribution of student preferences conditional on student value to otherwise vary arbitrarily. Since the distribution of student values is fixed according to $\eta$, we must have $\rho(V\times \mathcal{R})=\eta(V)$ for all $V\subseteq \RR.$

As we discuss in the next section, any economy has at least one stable matching. For a fixed $\DD$ and $C$, we let $E(\DD, C)$ be the set of all economies $(\DD, C, \vec{S}, \rho)$, hinting at the fact that our results depend only on the noise distribution and the total number of colleges.
\\
\\
We define two regimes of noise distributions $\DD$. $\DD$ is \vocab{$\beta$-max-concentrating} if and only if
\begin{equation}
    \Var[X^{(n)}] = O(n^{-\beta})
\end{equation}
for some $\beta > 0,$ where
$X^{(n)}$ is the maximum order statistic of $n$ samples from $\DD$ (i.e., the maximum of $n$ independent draws from $\DD$). For example, all bounded distributions are max-concentrating, and the Gaussian distribution is $1$-max-concentrating.\footnote{See, e.g., Proposition 4.7 of \cite{boucheron2012concentration}.} On the other hand, $\DD$ is \vocab{long-tailed} if and only if for all $d > 0$,
\begin{equation}
    \lim_{x\rightarrow \infty} \Pr_{X\sim \DD}[X > x + d\,|\,X > x] = 1.
\end{equation}
For example, the Pareto distribution is long-tailed, as are all subexponential distributions. These two regimes do not cover the set of all noise distributions; for example, the exponential and Gumbel distributions are neither max-concentrating nor long-tailed.

\subsection{Cutoff Characterization of Stable Matching}
We now specify the map between economies and stable matchings, using the cutoff characterization shown by \citet{azevedo2016supply}. Consider an economy $E = (\DD, C, \vec{S}, \rho) \in E(\DD,C).$

Let $\vec{P} = (P_1, P_2, \cdots, P_C)\in \RR^C$ be a vector of \vocab{cutoffs}, such that college $c$ has cutoff $P_c$. Student $\theta = (v, \succ)$ can \vocab{afford} college $c$ if and only if their estimated value at $c$ exceeds the cutoff $P_c$. Here, we recall that a student with value $v$ receives an estimated value $v + X_c$ at college $c$, where $X_1,\cdots,X_C$ are drawn i.i.d. according to $\DD$. The \vocab{demand} $D^\theta(\vec{P})$ of student $\theta$ at $\vec{P}$ is their most preferred firm among those they can afford; if they cannot afford any firm, $D^\theta(\vec{P}) = \emptyset$. In our setup, $D^\theta(\vec{P})$ is a random variable. Cutoffs $\vec{P}$ are \vocab{market-clearing} in the economy $E$ if and only if
\begin{equation}
    \int_{\Theta} \Pr\left[D^\theta(\vec{P})=c\right] \,d\rho(\theta) = S_c
\end{equation}
for all $c\in C$. In other words, the measure of students who demand a college is exactly equal to that college's capacity.\footnote{Capacities are exactly filled since the total capacity $S$ is less than the total mass of students $1$, and because every student prefers to be matched over being unmatched (i.e., colleges are overdemanded).} 

A random function $\mu:\Theta\rightarrow C\cup \{\emptyset\}$ is a \vocab{stable matching} of $E$ if and only if
\begin{equation}
    \mu(\theta) = D^\theta(\vec{P})
\end{equation}
for cutoffs $\vec{P}$ that are market-clearing in $E$. We refer the reader to Lemma 1 of \cite{azevedo2016supply} to see that this condition is in fact equivalent to the classical definition of stable matching, where stability here is with respect to college preferences based on the \textit{estimated} value of students.\footnote{This corresponds to a model in which colleges admit students based only on their estimated values (and do not strategize based on, for example, which students other colleges admit).}
\\
\\
Let $M(\DD, C)$ denote the set of all matchings that are stable for some economy $E\in E(\DD, C).$ Then any matching $\mu\in M(\DD, C)$ can be characterized by a set of market-clearing cutoffs $\vec{P}(\mu).$

\subsection{Preliminaries}
We now make some basic observations that will be useful for our results. We show that key properties of stable matchings operate at the level of student values rather than student types (which also encode student preferences). This means that in our analysis, we can focus on the fixed student value distribution $\eta$ while mostly ignoring the varying student type distribution $\rho$.

Consider a matching $\mu\in M(\DD, C)$, characterized by a vector of market-clearing cutoffs $\vec{P}=\vec{P}(\mu)$. Our results focus on the probability that a student $\theta=(v,\succ)$ matches to a college. This probability is equal to
\begin{equation}
    \Pr[\mu(\theta)\in C] = \Pr[v + X_c > P_c\text{ for some $c\in C$}],
\end{equation}
where $X_1, \cdots, X_C \simiid \DD.$ Notice that this probability does not depend on $\succ$, the student's preferences over colleges. Indeed, the probability that a student matches is exactly the probability that a student can afford some college; whether or not they can afford a college only depends on their value and the college's cutoff (not the student's preferences). Accordingly, we can drop $\theta$ and define
\begin{equation}
    p_\mu(v) := \Pr[v + X_c > P_c\text{ for some $c\in C$}],
\end{equation}
the probability that a student with value $v$ is matched under $\mu$. More generally, for a subset of colleges $C'\subseteq C$, we define
\begin{equation}
    p_\mu(v, C') := \Pr[v + X_c > P_c\text{ for some $c\in C'$}],
\end{equation}
the probability that a student with value $v$ can afford some college in $C'$. By definition, $p_\mu(v) = p_\mu(v, C).$

Notice that our analysis depends \textit{directly} only on three objects: $\eta$ (the distribution of students' true values), $\DD$ (the noise distribution), and $\vec{P}$ (the vector of cutoffs of each college). All other parameters forming an economy, such as the capacities of colleges and the distribution of student preferences, influence our analysis only insofar as they determine the cutoffs $P$. We only consider the interaction between these other parameters and these cutoffs $P$ on an as-needed basis. In this way, we may essentially take a set of cutoffs as given, assuming nothing else aside from the fact that they are market-clearing in some economy.

%% file: main-results.tex
\section{Main Results}\label{sec:results}
We now state our main results. Recall that there is a unique $v_S$ such that $\eta((v_S, \infty))=S$, meaning that the mass of students with true value above $v_S$ is equal to the total capacity $S$. In other words, if all students with value greater than $v_S$ match, this exactly meets supply; moreover, this is the ``noise-free'' outcome in which only the highest-value students match.

The first result pertains to economies with a max-concentrating noise distribution.

\begin{theorem}[Attenuation]\label{thm:attenuating}
    Let $\DD$ be max-concentrating and $v\in \RR$. Then for all $\epsilon > 0$, there exists $C(\epsilon)$ such that for all $C > C(\epsilon)$ and $\mu\in M(\DD, C)$,
    \begin{equation}
        p_\mu(v) < \epsilon \qquad \text{if $v < v_S$},
    \end{equation}
    and
    \begin{equation}
        p_\mu(v) > 1 - \epsilon \qquad \text{if $v > v_S$}.
    \end{equation}
\end{theorem}

\Cref{thm:attenuating} shows that when the number of colleges is large with max-concentrating noise, any student with value above $v_S$ will almost certainly be matched, whereas any student with value below $v_S$ will almost certainly not be matched. In this way, the market attenuates all noise (in a broad sense), since only the highest-value students are matched.

The next result pertains to economies with a long-tailed noise distribution.

\begin{theorem}[Amplification]\label{thm:amplifying}
    Let $\DD$ be long-tailed and $v\in \RR$. Then for all $\epsilon > 0$, there exists $C(\epsilon)$ such that for all $C > C(\epsilon)$ and $\mu\in M(\DD, C)$,
    \begin{equation}
        |p_\mu(v) - S| < \epsilon.
    \end{equation}
\end{theorem}

\Cref{thm:amplifying} shows that when the number of colleges is large with long-tailed noise, every student essentially has an equal probability $S$ of being matched. This limiting probability is independent of the student's value. In this sense, the market fully amplifies noise. Even if individual colleges rank students with some signal, the set of students who are matched is essentially random.
\\
\\
Observe the generality of both these results: for a fixed probability measure over student values $\eta$, total capacity $S$, and noise distribution $\DD$, these results hold for stable matchings of \textit{all} economies with a sufficiently large number of colleges that satisfy the specified regularity conditions.

%% file: proof-outline.tex
\section{Outline of Proofs}\label{sec:proof-outline}
In this section, we describe the main ideas in the proofs of \Cref{thm:attenuating,thm:amplifying}. Full proofs are given in \Cref{sec:attenuating,sec:amplifying}, respectively. In both proofs, we focus on the value
\begin{equation}
    p_\mu(v) := \Pr[v + X_c > P_c\text{ for some $c\in C$}],
\end{equation}
where $P_1, \cdots, P_C$ are the corresponding market-clearing cutoffs of a stable matching $\mu$. Conveniently, our proof strategy does not require calculating the cutoffs $P_1, \cdots, P_C$ for each possible economy; doing so would require individualized analysis for every distribution of student types and choice of college capacities. Rather, we simply take $P_1, \cdots, P_C$ as given, and---intermittently---use our knowledge that $P_1, \cdots, P_C$ are market-clearing. 

\paragraph{\Cref{thm:attenuating} Proof Outline.} With this in mind, we first consider \Cref{thm:attenuating}, the attenuating, ``wisdom of the crowds'' regime. To provide some intuition, we begin with the simplest case, $P_1 = P_2 = \cdots = P_C.$ Call this shared cutoff $P^*$. Then
\begin{align}
    p_\mu(v) &= \Pr[v + X_c > P^*\text{ for some $c\in C$}]\\
    &= \Pr[v + \max_{c\in C} X_c > P^*]\\
    &= \Pr[v > P^* - X^{(C)}],\label{eq:hippo}
\end{align}
where we recall that $X^{(C)}$ is the maximum order statistic of $\DD$. If $\DD$ is max-concentrating, then $X^{(C)}$ concentrates as $C$ grows large, meaning that \eqref{eq:hippo} is approximately $1$ when $v > P^* - \EE[X^{(C)}]$ and $0$ when $v < P^* - \EE[X^{(C)}].$ This example---though simple---makes clear an underlying intuition: that whether a student is matched is dictated roughly by their \textit{maximum} estimated value. (This is also the case analyzed in Theorem 1 of \citet{peng2023monoculture}.)

When the cutoffs $P_1, \cdots, P_C$ are not identical, our analysis becomes significantly more complicated. The quantity
\begin{equation}
    \Pr[v + X_c > P_c\text{ for some $c\in C$}]
\end{equation}
is not easily simplified in this case. It is, however, useful to borrow from the approach in the ``equal cutoffs'' case. In particular, we search for a ``dense cluster'' of colleges---that is, a large subset of cutoffs that are \textit{approximately} equal. More formally, we let $P^*$ be the smallest real number such that the interval $[P^*, P^* + \delta]$ (where we choose some $\delta=\delta(C)$) contains a ``large'' number of cutoffs, say at least $M=M(C)$.

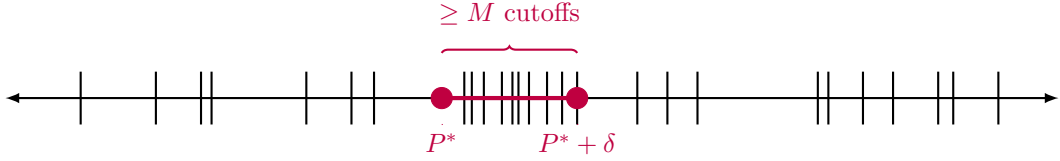
\begin{figure}[ht]
\vspace{0.5cm}
\begin{center}
\begin{tikzpicture}[scale=2]
\foreach \x in  {-3, -2.5, -2.2, -2.13, -1.5, -1.2, -1.05, -0.45, -0.4, -0.32, -0.2, -0.13, -0.09, -0.02, 0.1, 0.2, 0.3, 0.7, 0.9, 1.1, 1.9, 1.97, 2.2, 2.4, 2.7, 2.8, 3.1}
\draw[shift={(\x,0)},color=black, thick] (0pt,5pt) -- (0pt,-5pt);
\draw[shift={(-0.6,0)},color=purple] (0pt,-5pt) -- (0pt,-5pt) node[below] 
{$P^*$};
\draw[shift={(0.3,0)},color=purple] (0pt,-5pt) -- (0pt,-5pt) node[below] 
{$P^* + \delta$};
\draw[latex-latex, thick] (-3.5,0) -- (3.5,0) ;

\path [draw=purple, fill=purple] (-0.6,0) circle (2pt);
\path [draw=purple, fill=purple] (0.3,0.0) circle (2pt);
\draw[ultra thick, purple] (-0.6,0) -- (0.3,0);
\draw [shift={(0,0.3)}, decorate, decoration={brace}, purple, thick] (-0.6, 0) --  (0.3,0) node[midway,yshift=1.5em]{$\ge M$\text{ cutoffs}};
\end{tikzpicture}
\end{center}
\caption{$P^*$ is the smallest real number such that $[P^*, P^*+\delta]$ contains at least $M$ cutoffs.}
\end{figure}
We then consider two cases, delineated by a carefully chosen value $m=m(C)$:
\begin{itemize}
    \item Case 1: There are few colleges $(< m)$ with cutoff below $P^*$.
    \item Case 2: There are ``many'' colleges $(\ge m)$ with cutoff below $P^*$.
\end{itemize}
In the first case, the argument proceeds as follows. Since there are few colleges $(<m)$ with cutoff below $P^*$, we may ``ignore'' these colleges because they cannot contribute significantly to $p_\mu(v)$, the probability that a student matches. We then focus on the probability that a student can afford a college with cutoff at least $P^*$. Since there are many cutoffs in the narrow interval $[P^*, P^* + \delta],$ the probability that a student can afford \textit{any} college with cutoff at least $P^*$ is roughly the same as the probability that the student can afford a college with cutoff in the narrow interval. Moreover, since $\delta$ is small, this probability is approximately
\begin{equation}
    \Pr[v + X^{(M)} > P^*].
\end{equation}
Since $X^{(M)}$ concentrates when $M$ is large, this probability resembles a step function in $v$, as desired.

The argument in the second case is somewhat different; because there are many $(>m)$ colleges with low cutoffs, we cannot ``ignore'' all of these colleges and apply the same strategy as in the first case. Rather, we leverage the fact that there does not exist any dense cluster of cutoffs among the first $m$ cutoffs below $P^*$. By the pigeonhole principle, these cutoffs must be spread out: in particular, we are able to ensure that $P_m-P_1$ is significantly larger than $X^{(C)}$. We now employ the following reasoning. Now consider a student with value below $P_m - \EE[X^{(C)}].$ Then the probability that the student can afford a college with cutoff at least $P_m$ is approximately $0$, since even their maximum estimated value is unlikely to exceed $P_m$, let alone higher cutoffs. Meanwhile, the probability that a student matches with a college with cutoff less than $P_m$ must be small for almost all students, since the total capacity of these colleges is small (since $m$ is chosen to be small). Now consider a student with value greater than $P_m - \EE[X^{(C)}].$ Since $P_m - P_1$ is significantly larger than $P_m - \EE[X^{(C)}]$, it follows that the probability such a student can afford college $1$ is near $1$; therefore, the probability that the student can afford a college overall is near $1$. In this way, we have shown that $p_\mu(v)$ resembles a step function.

A key challenge in the proof is to select the appropriate notions of ``small'' and ``large,'' since these choices are sometimes competing. For example, we require a large number of colleges in a small interval in case 1, but this conflicts with our desire for a ``large gap'' in case 2. In this way, the dividing line between the two cases must be chosen in a way that satisfies both conditions. We leave the details of these compromises to the full proof in the appendix.

\paragraph{\Cref{thm:amplifying} Proof Outline.}
We now outline the proof of \Cref{thm:amplifying}, which is somewhat more straightforward. We begin with some high-level intuition for why long-tailed distributions induce this behavior. Consider two students with values $v_1$ and $v_2$. Suppose that $v_2 = v_1 - d$ for some $d > 0$. Now consider the likelihood that each applicant can afford a firm with cutoff $P$. These probabilities are $\Pr[v_1 + X > P]$ and $\Pr[v_2 + X > P]$ respectively. Consider the ratio
\begin{equation}
    \frac{\Pr[v_2 + X > P]}{\Pr[v_1 + X > P]} = \frac{\Pr[X > P - v_1 + d]}{\Pr[X > P - v_1]}.
\end{equation}
Then since $\DD$ is long-tailed,
\begin{equation}\label{eq:monstera}
    \lim_{P\rightarrow \infty} \frac{\Pr[X > P - v_1 + d]}{\Pr[X > P - v_1]} = 1.
\end{equation}
This tells us that if a college's cutoff is large, two students (with different values) can afford the college with approximately equal probability. To leverage this fact, we will show that almost all colleges must have large cutoffs. (Roughly speaking, this is true because if there were a large number of colleges with small cutoffs, then more students would be able to afford at least one of these firms than there is total capacity.) Then, as in the proof of \Cref{thm:attenuating}, we may ``ignore'' the small remaining number of colleges with low cutoffs. Note that since we require that the probabilities that two students can afford \textit{at least one} college to be similar, \eqref{eq:monstera} alone is insufficient; our actual analysis is more careful.

%% file: model-extended.tex
\section{Extension: Coalitions within a Broader Economy}
\label{sec:extended}

In our basic model, we considered economies in which all colleges share the same true preferences (but have different estimated preferences). We now extend our model so that we can analyze \textit{subsets} of colleges that share the same true preferences (but again have different estimated preferences). We call such subsets \textit{coalitions}.

We briefly sketch the structure of the extended model. In the extended model, there is a broader set of colleges $C^*$. Within this set of colleges, there is a subset of colleges $C\subseteq C^*$ which forms a coalition. This means that---as in the basic model---colleges in $C$ share the same true preferences over students determined by each student's true value $v$, but each forms rankings based on independent noisy draws of student values. Now, however, there are colleges in $C^*$ but not $C$ that may have arbitrary true and estimated preferences.

We then analyze the probability that a student with a given value can \textit{afford} a college in $C$, obtaining analogs to results in the basic model. Contrasting with our main results, the results in this section focus on the probability a student can \textit{afford} a college rather than the probability a student \textit{matches} with a college. This is because whether or not a student matches with a college in a coalition depends significantly on their preferences with respect to colleges outside of the coalition.

Also distinct from results in the basic model, our extended results will refer to a subset of colleges $C'\subseteq C$ within the coalition $C\subseteq C^*,$ such that $|C'|/|C| \approx 1.$ This means that our results in this section are somewhat weaker than what is shown in the basic model. On the other hand, these (somewhat) weaker results can be shown with fewer assumptions. 

As in the analysis of our basic model, our analysis will leverage the cutoff structure of stable matching. Again, we will largely be able to take cutoffs as given, and use market-clearing conditions indirectly. Furthermore, our analysis will depend essentially only on the cutoffs within a coalition. In this way, the key pieces of analysis from the basic model will port over directly.

\subsection{Model and Definitions}

\paragraph{General noisy economy.}
There is a continuum of students of measure 1 to be matched with a finite set of colleges $C^* = \{1, 2, \cdots, C^*\}.$ Each student is identified by their \textit{true type} $\theta = (\bm{v}, \succ)$, where
\begin{equation}
    \bm{v} = (v_1, v_2, \cdots, v_{C^*}) \in \RR^C
\end{equation}
are the \textit{true values} of the student. Here, $v_c$ is the true value of the student at college $c$. In contrast to the basic model, students have multiple true values, which differ across colleges. 

A student with true values $\bm{v}$ has \textit{approximate values} $\hat{\bm{v}},$ where
\begin{equation}
    \hat{\bm{v}} = (\hat{v}_1, \hat{v}_2, \cdots, \hat{v}_{C^*})
\end{equation}
is a random variable over $\RR^{C^*}$. ($\hat{v}_1, \hat{v}_2, \cdots, \hat{v}_{C^*}$ need not be independent.) Thus, a student $(\bm{v},\succ)$ has true value at college $c$ equal to $v_c$, and the college obtains an approximation of the student's value equal to $\hat{v}_c.$

An \textit{economy} is given by an ordered pair $(\rho, \vec{S}, \hat{\bm{V}}),$ where $\rho$ is a probability measure over the set of student types $\Theta = \RR^{C^*}\times \mathcal{R}$ and $\vec{S}$ is a vector of college capacities. As before, we hold the total capacity of colleges to be a fixed constant $S\in (0,1)$, and assume that all students prefer being matched to any college over being unmatched. $\hat{\bm{V}} = \{\hat{\bm{v}}\}_{\bm{v}\in \RR^{C^*}}$ is a set of independent random variables over $\RR^{C^*}$, where $\hat{\bm{v}}$ give the approximate values of a student with true values $\bm{v}$. For example, one could choose $\hat{\bm{V}}$ such that $\hat{\bm{v}}$ is equal to $\bm{v}$ perturbed by noise drawn from a specified multivariate normal distribution for all $\bm{v}\in \RR^C$. Here, $\hat{\bm{V}}$ generalizes the role of the noise distribution $\DD$ in the basic model, by allowing for more broader relationships between true values and approximate values.

(Formally, we require $\hat{\bm{V}}$ to be chosen in a way that induces a valid probability distribution over true values and approximate values, where $\rho$ specifies the marginal distribution of true values, and $\hat{\bm{v}}$ specifies the conditional distribution for each true value.)

\paragraph{Cutoff characterization of stable matching.} 
Consider an economy $E=(\rho, \vec{S}, \hat{\bm{V}}).$ Given cutoffs $\vec{P} = (P_1, \cdots, P_{C^*}),$ we have that a student with true values $\bm{v}$ can afford college $c$ if and only if $\hat{v}_c > P_c$. A student's demand $D^{\theta}(\vec{P})$ at cutoffs $\vec{P}$ is equal to the college they most prefer among those they can afford; $D^{\theta}(\vec{P})=\emptyset$ if the student cannot afford any college. Note that $D^{\theta}(\vec{P})$ is a random variable. Then, as in the basic model, the cutoffs $\vec{P}$ are market-clearing if and only if
\begin{equation}
    \int_{\theta} \Pr[D^{\theta}(\vec{P})=c]\,d\rho(\theta) = S_c
\end{equation}
for all $c\in {C^*}$.\footnote{Capacities are exactly filled since the total capacity $S$ is less than the total measure of students $1$, and because every student prefers to be matched over being unmatched.} If $\vec{P}$ is market-clearing, then the function $\mu:\theta\mapsto D^{\theta}(\vec{P})$ is a stable matching.

\paragraph{College coalitions.}
We are interested in analyzing the behavior of subsets of colleges in an economy that share the same true preferences over students, but who each make noisy decisions. We now formally define these subsets.

In an economy $E=(\rho, \vec{S}, \hat{\bm{V}})$, a subset of colleges $C\subseteq {C^*}$ is a $(\DD,\eta)$\vocab{-coalition} if and only if the following conditions are met:
\begin{enumerate}
    \item The true value of a student is same at each college in $C$.\footnote{Formally, the $\rho$-measure of students who have different true values at two colleges in $C$ is equal to $0$: 
    \begin{equation}
        \rho\left(\{(\bm{v}, \succ) \in \Theta: \text{exists }c,c'\in C\text{ such that }v_c\neq v_c'\}\right) = 0.
    \end{equation}
    }
    If this condition is satisfied, for a student with true values $\bm{v}$, we simply use $v$ to refer to the common true value of the student at colleges in $C$. In other words, $v_c=v$ for all $c\in C$. We call $v$ the student's \textit{true value at $C$}.
    \item The approximate values of a student at colleges in $C$ is equal (in distribution) to their true value at $C$ perturbed by independent noise drawn from $\DD$: For all $\bm{v}\in \RR^{C^*},$
    \begin{equation}
        \hat{v}_c \eqdist v + X_c
    \end{equation}
    for all $c\in C$,
    where $X_1, X_2, \cdots, X_C$ are drawn i.i.d. according to $\DD$.
    \item The distribution of true values at $C$ is given by the probability measure $\eta$. That is,
    \begin{equation}
        \rho((\bm{v},\succ): v\in C) = \eta(V)
    \end{equation}
    for all $V\subseteq \RR.$
\end{enumerate}

In essence, a subset of colleges form a coalition if they behave like the full set of colleges in the basic model---i.e., the colleges in the coalition have the same true value for each student and each obtain independent noisy estimates of this value. (But, as in the basic model, students may have arbitrary preferences over colleges in the coalition, and colleges may differ in their capacities.) Indeed, one may check that in the basic model, the entire set of colleges $C$ forms a $(\DD,\eta)$-coalition.\footnote{The reader might observe that the conditions for a coalition are somewhat simpler than in the basic model: the constants $\alpha$ and $\gamma$ governing regularity conditions are notably absent in the extended model. The reason is that these regularity conditions are not necessary for the somewhat weaker results we establish in the extended model. On the other hand, the regularity conditions are also not sufficient to establish stronger results in the extended model, as we explain below.} As we will see, coalitions satisfy very similar properties to the full set of colleges in the basic model.

\paragraph{Preliminaries.}
Consider an economy $E = (\rho, \vec{S}, \hat{\bm{V}})$ and a $(\DD, \eta)$-coalition $C$. Let $\mu$ be a stable matching in $E$ with corresponding market-clearing cutoffs $\vec{P}$. Our results focus on the probability that a student $\theta=(\bm{v},\succ)$ can \textit{afford} a college in $C'\subseteq C$. Recall that the student's true value at $C$ is denoted simply by $v$. Then the probability the student can afford a college in $C'$ is
\begin{equation}
    \Pr[v + X_c > P_c\text{ for some $c\in C'$}],
\end{equation}
which does not depend on the true values of the student at colleges outside of the coalition, nor the preferences of the student. We set
\begin{equation}
    p_\mu(v, C') := \Pr[v + X_c > P_c\text{ for some $c\in C'$}]
\end{equation}
to denote the probability that a student with true value $v$ at coalition $C$ can afford a college in $C'\subseteq C$. In this way, we can begin to see how the setup of the extended model essentially reduces to that of the basic model, since the quantity we are interested in depends only on the true value of a student at the coalition $C$.
Given that the true values of a student at colleges outside of $C$ are irrelevant, we will---when unambiguous---refer to the true value of a student at the coalition $C$ as simply the true value of the student.

\subsection{Results}
We now state analogues of \Cref{thm:attenuating} and \Cref{thm:amplifying} in the extended model. 

\begin{theorem}[Attenuation in Coalitions]\label{thm:attenuating-extended}
    Let $\DD$ be max-concentrating. Then for any $\epsilon > 0$, there exists $C(\epsilon)$ such that the following holds.
    Let $C$ be a $(\DD, \eta)$-coalition in an economy $E$ with stable matching $\mu$ such that $|C|>C(\epsilon)$. Then there exists $v'\in \RR$ and $C'\subseteq C$ with $\frac{|C'|}{|C|} > 1 - \epsilon$, such that
    \begin{equation}
        p_\mu(v, C') < \epsilon \qquad\text{ for all $v<v'-\epsilon$},
    \end{equation}
    and
    \begin{equation}
        p_\mu(v, C) > 1 - \epsilon \qquad\text{ for all $v>v'+\epsilon$}.
    \end{equation}
\end{theorem}

\Cref{thm:attenuating-extended} shows that there is a value $v'$ at which there is a sharp change in behavior: For all students with true value $v<v'$, the probability that they can afford a college in $C'\subseteq C$ is very small, where $C'$ contains almost all of the colleges in $C$. Meanwhile, for all students with true value $v>v'$, the probability that they can afford a college in $C$ is almost $1$.

The result also implies a sort of ``approximately no justified envy'' property in noisy matching markets when $\DD$ is max-concentrating: the pairs of students with values $v_1, v_2$ such that $v_1 > v_2$ but where the student with value $v_2$ can afford a college in $C'$ while the student with value $v_1$ cannot afford a college in $C$ is vanishing in measure.

\begin{theorem}[Amplification in Coalitions]\label{thm:amplifying-extended}
    Let $\DD$ be long-tailed. Then for any $\epsilon > 0$, there exists $C(\epsilon)$ such that the following holds.
    Let $C$ be a $(\DD, \eta)$-coalition in an economy $E$ with stable matching $\mu$ such that $|C|>C(\epsilon)$. Then there exists $S'\in \RR$ and $C'\subseteq C$ with $\frac{|C'|}{|C|} > 1 - \epsilon$, such that
    \begin{equation}
        |p_\mu(v, C')-S'| < \epsilon
    \end{equation}
    for all $v$ except on a set of $\eta$-measure at most $\epsilon$.
\end{theorem}

\Cref{thm:amplifying-extended} shows that there exists a large subset $C'\subseteq C$ such that the probability of affording a college in $C'$ is essentially the same for all students, regardless of a student's true value at the coalition.\footnote{A natural question is if the restriction to a large subset $C'\subseteq C$ (absent in \Cref{thm:attenuating} and \Cref{thm:amplifying}) is necessary for \Cref{thm:attenuating-extended} and \Cref{thm:amplifying-extended}. To see why it is, consider the following example. We focus on the attenuating regime, but the same reasoning extends to the amplifying regime. Suppose that all colleges share the same true values of students and that students have shared preferences over these colleges. Then, colleges can be ranked in strict order of popularity. Now consider a coalition that consists of all the most popular colleges and the single least popular college. Then the set of students who can afford a college in the coalition is essentially the same as the set of students who can afford the single least popular college (which has a very low cutoff in comparison to other colleges in the coalition). Since decisions by a single college are noisy, we have to ``remove'' this college from the coalition to establish the desired result. The result shows, however, that at most a vanishing proportion of colleges need to be removed from a coalition to obtain the desired behavior. 

This example raises the question of why such a restriction is \textit{not} required in the basic model. The reason is that \Cref{thm:attenuating} and \Cref{thm:amplifying} pertain to the probability that a student \textit{matches} to a college in $C$: in such a case, a vanishing subset of colleges does not have to be removed, since the number of students matched to these colleges is also vanishing. This cannot be leveraged when analyzing the probability that a student can \textit{afford} a college in $C$. This perhaps raises yet another question, of why our analysis in the extended model focuses on the probability a student can afford rather than match with a college in the coalition $C$. The reason is that the matching behavior can be altered significantly by changing the distribution of student preferences with respect to colleges outside of the coalition. However, by imposing stricter assumptions on the structure of preferences---e.g., students strictly prefer colleges in the coalition to those outside of the coalition---it is likely that results on the matching behavior can be obtained in the extended setting.}

\Cref{thm:attenuating-extended} and \Cref{thm:amplifying-extended} are especially powerful when applied to many coalitions at once. For example, one may consider an economy comprised of colleges drawn from a finite number of types, such that students do not distinguish between colleges of the same type. Then if the colleges of each given type form a coalition,  \Cref{thm:attenuating-extended} suggests that under max-concentrating noise, large economies will behave as if without noise, both in terms of who is matched and to where (which type) each student is matched. Meanwhile, \Cref{thm:amplifying-extended} suggests that under long-tailed noise, large economies will not sort students based on their value to colleges. In other words, when additional structure is assumed in an economy, \Cref{thm:attenuating-extended} and \Cref{thm:amplifying-extended} imply results about the overall outcome of an economy.

Full proofs of \Cref{thm:attenuating-extended,thm:amplifying-extended} are given in \Cref{sec:proofs-extended}, and follow essentially from the analysis given for \Cref{thm:attenuating,thm:amplifying}.

\begin{figure}
    \begin{center}
        \includegraphics[width=\linewidth]{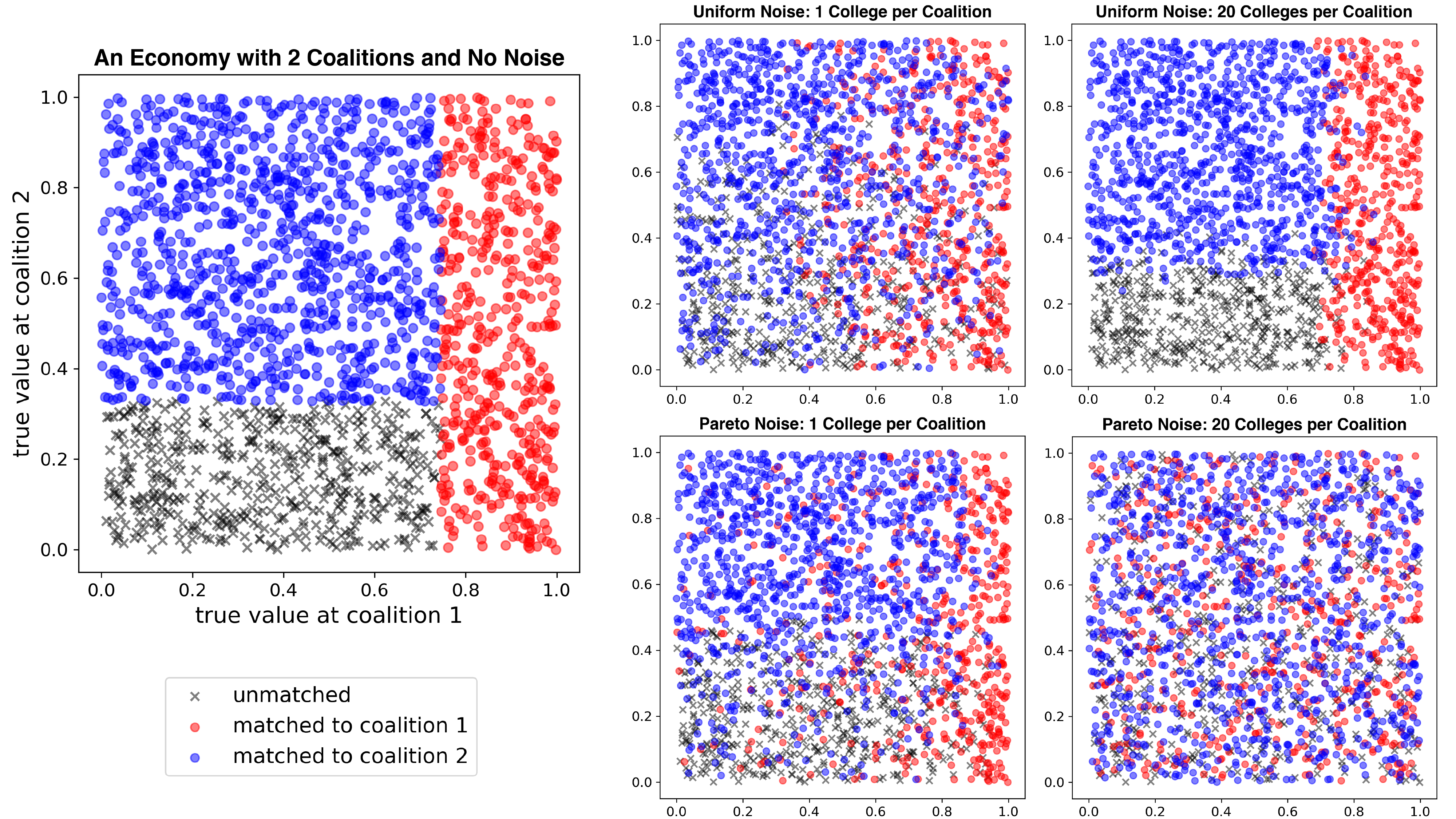}
    \end{center}
    \caption{An economy with 2000 students and two coalitions of colleges. Coalition $1$ (red) has total capacity $500$ and coalition $2$ (blue) has total capacity $1000$. All students strictly prefer colleges in coalition $1$ to those in coalition $2$. Thus, in the noiseless case (left), the students with the highest true values at coalition 1 ($v_1>\frac{3}{4}$) match to a college in coalition 1, those out of the remaining students with the highest true values at coalition 2 ($v_2 > \frac{1}{3}$) match to a college in coalition 2, and the remaining students are unmatched. The top row on the right illustrates \Cref{thm:attenuating-extended}: as coalition size grows large with a max-concentrating distribution, the market approaches the noiseless setting. The bottom row illustrates \Cref{thm:amplifying-extended}: with a long-tailed distribution, noise is amplified as coalition size increases. In particular, the bottom right matching is \textit{noisier} than the bottom left matching, and is approaching uniform at random matching.}
    \label{fig:two-coalitions}
\end{figure}

\subsection{Computational Experiment: Noisy Economy with Two Coalitions}

We now illustrate our results using simulations. Consider an economy with colleges $1, 2, \cdots, 2C$. There are two coalitions $\CC_1 = \{1, 2, \cdots, C\}$ and $\CC_2 = \{C + 1, C + 2, \cdots, 2C\}.$ Therefore, we may think of each student as having two true values $v_1$ and $v_2$, such that colleges $c \in \CC_1$ rank students based on a noisy estimate $v_1 + X_{c,1}$ and colleges $c\in \CC_2$ rank students based on a noisy estimate $v_2 + X_{c,2}.$ Here, $X_{c,1}$ and $X_{c,2}$ are drawn i.i.d. from a noise distribution $\DD$. Let the true values of students $(v_1, v_2)$ be distributed uniformly on $[0,1]\times [0,1],$ and assume that students have preferences as follows: each student prefers all colleges in $\CC_1$ to those in $\CC_2$, but prefers colleges within a coalition uniformly at random. We also assume that colleges in coalition $\CC_1$ each have capacity $0.25/C$ and colleges in coalition $\CC_2$ each have capacity $0.5/C$, meaning that the total capacity of $\CC_1$ is $0.25$ and the total capacity of $\CC_2$ is $0.5$. This essentially ``normalizes'' the economy according to $C$.

First consider what would happen in such an economy without noise. We calculate what the cutoffs are in this economy. Intuitively, the cutoffs within each coalition are the same by symmetry. Denote these cutoffs by $P_1^*$ and $P_2^*$ respectively. Since all students prefer colleges in $\CC_1$ to $\CC_2$, the measure of students who match to a college in $\CC_1$ is equal to $1 - P_1^*$; since the total capacity of $\CC_1$ is $0.25$, this implies that $P_1^* = 3/4.$ The measure of students who are matched to a college in $\CC_2$ is thus $\frac{3}{4}(1 - P_2^*)$. Since the total capacity of $\CC_2$ is $0.5$, this implies $P_2^* = 1/3.$

We compare the outcomes of noisy economies to this ``noise-free'' outcome. We run simulations with $2000$ students with values $(v_1, v_2)$ drawn uniformly at random on $[0,1]\times [0,1]$. We consider two choices of the number of colleges $C$, $1$ and $20$, and two choices of the noise distribution $\DD$, the uniform $[0,1]$ distribution and the Pareto distribution with shape $2$ and scale $0.3$. These distributions are chosen such that in the case $C=1$, the resulting outcomes are similarly noisy. Outcomes are plotted in \Cref{fig:two-coalitions}. The main observation is that when $C=20$, the outcome is approximately noise-free under the uniform distribution, reflecting \Cref{thm:attenuating-extended}; meanwhile the outcome appears entirely noisy under the Pareto distribution, reflecting \Cref{thm:amplifying-extended}.

%% file: discussion.tex
\section{An Approach to Study Imperfect Preference Formation}\label{sec:discussion}
We now take a step back, abstracting away the details of our analysis and the specific questions we answer. We suggest that the general approach we have taken can be adapted to answer a much broader class of questions about the effects of imperfect preference formation in matching contexts. 

Stable matching has proven to be a useful solution concept to understand how preferences induce matchings when participants are embedded in a broader market. However, the preference formation process is necessarily imperfect. Participants in a market have limited information on which to form their preferences---and even given this information, participants may form preferences in a noisy or biased way. Our approach is to start with a set of true preferences and then \textit{exogenously} specify how these true preferences map onto approximate preferences. This map, denoted by $\mathcal{A}$ in the diagram below, models the imperfect preference formation process. In the present work, for example, $\mathcal{A}$ takes the true preferences of colleges and perturbs these preferences by independent noise. A large body of work studies such imperfect preference formation processes.
\begin{figure}[tbh]
\centering
\begin{tikzcd}
\text{true preferences} \arrow [dr,dashed, bend right,"\pi \circ \mathcal{A}"] \arrow[r,"\mathcal{A}"] &
  \text{approximate preferences} \arrow[d,"\pi"]\\
& \text{matching}
\end{tikzcd}
\caption{A framework for studying the effect of imperfect preference formation.}
\end{figure}
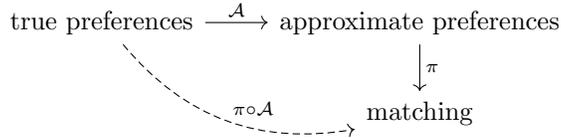

 However, we then specify a map $\pi$ from approximate preferences to a matching outcome; a natural choice of $\pi$ maps preferences to a stable matching. We study the induced mapping from true preferences to matching outcomes given (intuitively) by $\pi\circ \mathcal{A}$. Beginning with a set of true preferences allows us to compare outcomes under different maps $\AA$; in particular, it is useful to take $\mathcal{A}$ to be the identity map as a baseline. In the present work, we saw how the distributional properties of the noise can drastically affect the outcome of the market. More generally, our findings illustrate that the interaction between $\mathcal{A}$ (the preference formation process) and $\pi$ (stable matching) is not necessarily clear \textit{a priori}---for example, exhibiting emergent behavior in large markets.

We briefly suggest some other questions that may be (further) answered using this framework:
\begin{itemize}
    \item What happens when some firms are biased towards a group of job candidates? \citep{bertrand2004emily, quillian2017meta}. Is this bias amplified by the market, or does the presence of unbiased firms correct this bias? One may also ask analogous questions, replacing bias with statistical discrimination \citep{phelps1972statistical}.
    \item What happens when all firms use the same algorithm to evaluate applicants? \citep{kleinberg2021algorithmic}. In the resulting matching outcome, does this harm or benefit applicants and firms? \citet{peng2023monoculture} applies the framework to this question.
    \item What happens when each job candidate only applies to a subset of the full set of firms? \citep{immorlica2003marriage, haeringer2009constrained, calsamiglia2010constrained}. Can markets with short preference lists approximate a market with full preference lists?
    \item What happens when students use common college rankings to approximate their preferences, as compared to conducting independent research on colleges? Would students benefit more from accurate information on shared preference components (e.g., graduation rates) or idiosyncratic ones (e.g., a specialty program)?
    \item What happens when the above characteristics apply to both sides of the market simultaneously? Do the same effects occur, or does, e.g., two-sided noise amplify in ways that one-sided noise does not?
\end{itemize}

A few more remarks: The decomposition of the matching process into two components $\AA$ and $\pi$ is useful both methodologically and in terms of tractability. Methodologically, the decomposition allows us to cleanly isolate and compare the effect of imperfect preference formation $\AA$. Exogenizing the preference formation process is useful in this way. At the same time, however, it limits the analysis of subtler interactions such as the continual updating and strategic use of information during the matching process, as is studied in prior work discussed in \Cref{sec:related-work}. 

We further note that any stable matching model may be used to model $\pi$. In the present work, we use the continuum model of \citet{azevedo2016supply}, though $\pi$ may also be analyzed, for example, using the standard finite model \citep{gale1962college} or a finite capacity continuum model \citep{arnosti2022continuum}. %

%% file: conclusion.tex
\section{Conclusion}\label{sec:conclusion}
In this paper, we demonstrated that matching markets can serve as both strong noise attenuators and strong noise amplifiers, all without colleges sharing information. These results reveal that noisy behavior at the individual level can coalesce to produce striking and non-obvious market-level behavior. The results also reflect an underlying observation: that market outcomes depend on \textit{extremes}: what is the \textit{most favorable} evaluation a student receives from a college?

A few interesting open questions remain in the setup of our model. First, one might ask about the behavior that arises from noise distributions $\DD$ lying between the two regimes of our focus. We focused on these two regimes, since they are fairly broad, and induce the most extreme effects. Precise characterizations of intermediate behavior would be interesting. Second, one might take the results we begin to develop in the extended model (\Cref{sec:extended}) and generalize further. Is it possible, for example, to identify noise attenuation and amplification in a model where all participants (on both sides of the market) have idiosyncratic true preferences?

Finally, as discussed in \Cref{sec:discussion}, we suggest that the approach taken here may find broader applicability in studying the effect of imperfect preference formation at the market level. 

%% file: proof-attenuating.tex
\section{Proof of Theorem \ref{thm:attenuating} (Full Attenuation)}\label{sec:attenuating}

Consider a student with value $v.$ We would like to show that for all $\epsilon > 0$, there exists $C(\epsilon)$ such that for all $C > C(\epsilon)$ and $\mu\in M(\DD, C)$,
    \begin{equation}
        p_\mu(v) < \epsilon
    \end{equation}
    if $v < v_S$, and
    \begin{equation}
        p_\mu(v) > 1 - \epsilon
    \end{equation}
    if $v > v_S.$

First observe that it suffices to show that for all $\epsilon > 0$, there exists $C(\epsilon)$ such that for all $C > C(\epsilon)$ and $\mu\in M(\DD, C)$,
\begin{equation}
    \int_{v_S}^\infty p_\mu(v)\,d\eta(v) < \epsilon^2,
\end{equation}
meaning that the mass of matched students with true value below $v_S$ vanishes as the number of colleges grows large.\footnote{Indeed, suppose the statement is true for $v < v_S$. Then for $\epsilon$ sufficiently small, $v < v^*$ for $v^*$ satisfying $\eta((v^*, v_S))=\epsilon$. Then for all $C > C(\epsilon)$ and $\mu\in M(\DD,C)$, we have that
\begin{align}
    \epsilon^2 &> \int_{-\infty}^{v_S} p_\mu(v)\,d\eta(v)\\
    &> \int_{v^*}^{v_S} p_\mu(v)\,d\eta(v) \\
    &\ge \eta((v^*, v_S)) p_\mu(v)\\
    &\ge \epsilon p_\mu(v),
\end{align}
so $p_\mu(v) < \epsilon$ as desired. The case for $v > v_S$ is analogous.} Therefore, \Cref{thm:attenuating} follows from the following proposition.
\begin{proposition}\label{thm1v2}
    Let $\DD$ be $\beta$-max-concentrating and $v\in \RR$. Let $\mu \in M(\DD, C)$. Then
    \begin{equation}
        \int_{v_S}^\infty p_\mu(v)\,d\eta(v) = O\left(C^{-K(\beta,\gamma)}\right),
    \end{equation}
    where
    \begin{equation}
        K(\beta,\gamma) := \frac{2\beta\gamma}{3\beta\gamma + 2\beta + 5\gamma + 6}.
    \end{equation}
\end{proposition}
(Recall that $\gamma>0$ is a constant such that $\eta([x,x+\delta]) = O(\delta^{\gamma}).$) In other words, \Cref{thm1v2} shows that the mass of students who match with value less than $v_S$ vanishes as $C$ grows large. \Cref{thm1v2} can also be viewed as an approximate ``rural hospital theorem,'' in the sense that the set of students who are matched (or unmatched) is nearly deterministic. We spend the remainder of this section proving \Cref{thm1v2}.

\subsection*{Proof of \Cref{thm1v2}}
Consider any stable matching $\mu\in M(\DD, C)$, so that $C=\{1,2,\cdots,C\}$ is the set of firms. Let $\vec{P}=\vec{P}(\mu)$ be the corresponding vector of market-clearing cutoffs. Without loss of generality, we may assume $P_1\le P_2\le \cdots \le P_C.$ A remarkable feature of the proof is that we need not explicitly determine these cutoffs. Instead, we may analyze them as given, using the knowledge that they are market-clearing.

Define the constants
\begin{align}
    \phi_1 := - \frac{2\beta}{3\beta\gamma + 2\beta + 5\gamma + 6},
    \quad \phi_2 := \frac{5\gamma + 6}{3\beta\gamma + 2\beta + 5\gamma + 6},
    \quad \phi_3 := 1 - \frac{2\beta\gamma}{3\beta\gamma + 2\beta + 5\gamma + 6}.
\end{align}
We will perform most of our analysis in terms of $\phi_1, \phi_2,$ and $\phi_3$, and then substitute in these values at the end.\footnote{This is also the approach we used to find these ``magical'' constants, by first computing various ``error'' terms in terms of $\phi_1, \phi_2,$ and $\phi_3$ can choosing these constants to minimize the total error.}

The key idea of the proof is to identify a ``lowest dense interval'' of cutoffs. Set $\delta = C^{\phi_1}.$ Define $P^*$ to be the smallest real number such that the interval $[P^*, P^*+\delta]$ contains at least $C^{\phi_2}$ cutoffs. If such an interval does not exist, then we set $P^* := \infty.$

\begin{figure}[h]
\vspace{0.5cm}
\begin{center}
\begin{tikzpicture}[scale=2]
\foreach \x in  {-3, -2.5, -2.2, -2.13, -1.5, -1.2, -1.05, -0.45, -0.4, -0.32, -0.2, -0.13, -0.09, -0.02, 0.1, 0.2, 0.3, 0.7, 0.9, 1.1, 1.9, 1.97, 2.2, 2.4, 2.7, 2.8, 3.1}
\draw[shift={(\x,0)},color=black, thick] (0pt,5pt) -- (0pt,-5pt);
\draw[shift={(-0.6,0)},color=purple] (0pt,-5pt) -- (0pt,-5pt) node[below] 
{$P^*$};
\draw[shift={(0.3,0)},color=purple] (0pt,-5pt) -- (0pt,-5pt) node[below] 
{$P^* + \delta$};
\draw[latex-latex, thick] (-3.5,0) -- (3.5,0) ;

\path [draw=purple, fill=purple] (-0.6,0) circle (2pt);
\path [draw=purple, fill=purple] (0.3,0.0) circle (2pt);
\draw[ultra thick, purple] (-0.6,0) -- (0.3,0);
\draw [shift={(0,0.3)}, decorate, decoration={brace}, purple, thick] (-0.6, 0) --  (0.3,0) node[midway,yshift=1.5em]{$\ge C^{\phi_2}$\text{ cutoffs}};
\end{tikzpicture}
\end{center}
\caption{$P^*$ is the smallest real number such that $[P^*, P^*+\delta]$ contains at least $C^{\phi_2}$ cutoffs.}
\end{figure}
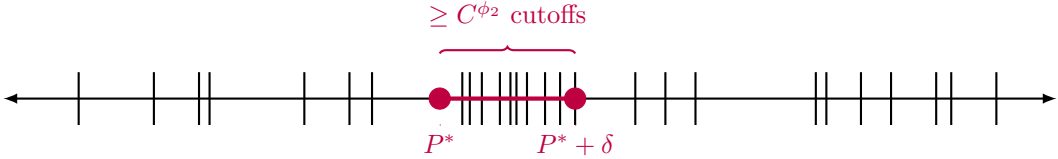

We use the interval $[P^*, P^*+\delta]$ to proceed in two cases, depending on the number of cutoffs below $P^*.$ For an interval $I$, define $C(I):=\{c\in C: P_c\in I\}$ to be the set of colleges with cutoffs
contained in $I$. Then the two cases are as follows: 
\begin{itemize}
    \item \textbf{Case 1.} Few cutoffs below $P^*$:
    \begin{equation}
        |C((-\infty, P^*))| < C^{\phi_3}.
    \end{equation}
    
    \item \textbf{Case 2.} Many cutoffs below $P^*$:
    \begin{equation}
        |C((-\infty, P^*))| \ge C^{\phi_3}.
    \end{equation}
\end{itemize}
We choose these two cases for the following reason. The first case guarantees the existence of a dense set of cutoffs ``early on'' in the sequence of colleges. The second case, does not have a dense set of cutoffs early on. By the pigeonhole principle, this instead implies that there must exist a large gap between cutoffs early on. As we will see, both the existence of an ``early'' dense set or large gap is sufficient to obtain the desired result.

\subsubsection*{Case 1: Few Cutoffs Below $P^*$}

At a high level, our analysis in this case proceeds as follows. Set
\begin{align}
    \CC_1 &:= C((-\infty, P^*))\\
    \CC_2 &:= C((P^*, \infty)).
\end{align}
First, because there are few cutoffs below $P^*$, the mass of students who match to colleges in $\CC_1$ is small---i.e., going to $0$ as $C\rightarrow \infty$. Therefore, we can focus on the probability a student can afford a college in $\CC_2.$ This probability is approximately a step function as for $C$ large, since the dense cluster of colleges in $[P^*, P^*+\delta]$ contains a large number of colleges, serving as an ``arbitrarily strong filter.''
\\
\\
In accordance with this plan, observe that
\begin{equation}
    \int_{v_S}^\infty p_\mu(v)\,d\eta(v) \le S(\CC_1) + \int_{v_S}^\infty p_\mu(v,\CC_2)\,d\eta(v),
\end{equation}
since the total measure of students with value below $v_S$ who match to a college is at most the total measure $S(\CC_1)$ of students (of any value) who match with a college in $\CC_1$ plus the total measure of students with value below $v_S$ who can afford a college in $\CC_2.$

We proceed by bounding the two components of the RHS. Bounding $S(\CC_1)$ is straightforward: there are at most $C^{\phi_3}$ colleges in this set, each of which has capacity at most $\frac{\alpha}{C}.$ This gives the following bound.
\begin{proposition}\label{prop:duck-1}
\begin{equation}
    S(\CC_1) \le \alpha C^{\phi_3 - 1} = O\left(C^{-K(\beta,\gamma)}\right).
\end{equation}
\end{proposition}
\begin{proof}
    The capacity of each college is at most $\frac{\alpha}{C}$. Since there are fewer than $C^{\phi_3}$ colleges in $C((-\infty, P^*))$, the mass of students who match to these colleges is less than $C^{\phi_3}\cdot \frac{\alpha}{C} = \alpha C^{\phi_3 - 1}.$ The result follows since
    \begin{equation}
        \phi_3 - 1 = -\frac{2\beta\gamma}{3\beta\gamma + 2\beta + 5\gamma + 6} = K(\beta,\gamma).
    \end{equation}
\end{proof}
Bounding the second component---the measure of students with value below $v_S$ who match with a college in $C([P^*,\infty))$---is more involved. We start by showing that $p_\mu(v, \CC_2)$, the probability that a student with value $v$ can afford a college in $\CC_2$, resembles a step function in $v$ that jumps at $P^* - \EE[X^{(C^{\phi_2})}],$ where we recall that $\phi_2$ is the number of colleges in the dense cluster. We can use this result to analyze the second integral, since if $p_\mu(v, \CC_2)$ resembles a step function, then this step must occur at roughly $v_S$ or otherwise too many (more than $S$) or too few (at most $S$) students will be able to afford a college. In turn, $p_\mu(v, \CC_2)$ must be small for almost all students with value at most $v_S$.

\begin{proposition}\label{prop:tomato}
\quad
    \begin{itemize}
        \item[(i)] If $v > P^* - \EE[X^{(C^{\phi_2})}]  + 2\delta,$
        \begin{equation}
            p_\mu(v, \CC_2) \ge 1 - O\left(C^{-2\phi_1 - \beta\phi_2}\right).
        \end{equation}
        \item[(ii)] If $v < P^* - \EE[X^{(C^{\phi_2})}]  - \delta,$
        \begin{equation}
            p_\mu(v, \CC_2) \le O\left(C^{1-2\phi_1 - (1 + \beta)\phi_2}\right).
        \end{equation}
    \end{itemize}
\end{proposition}

We first sketch the proof strategy and intuition for \Cref{prop:tomato}. For part (i), we will lower bound $p_\mu(v, \CC_2)$ by the probability that a student with value $v$ can afford a college in the dense cluster $C([P^*, P^* + \delta]).$ In turn, this can be lower bounded by the probability that among the at least $C^{\phi_2}$ scores $\theta$ receives from the colleges in the dense cluster, at least one exceeds $P^* + \delta$, i.e., $\Pr[v + X^{(C^{\phi_2})} > P^* + \delta$]. This probability is close to $1$, since $v + \EE[X^{(C^{\phi_2})}] > P^* + 2\delta$ and $X^{(C^{\phi_2})}$ concentrates as $C$ grows large (due to the $\beta$-max-concentrating property).

Our approach for part (ii) is similar, but uses an additional observation. We can upper bound $p_\mu(v, \CC_2)$ by the probability that among the at most $C$ scores $v$ receives from the colleges in $C([P^*,\infty)),$ at least one exceeds $P^*$---i.e., $\Pr[v + X^{(C)} > P^*].$ The key observation is that this probability is at most $C^{1 - \phi_2}\Pr[v + X^{(C^{\phi_2})} > P^*]$ by a union bound. $\Pr[v + X^{(C^{\phi_2})} > P^*]$ is small since $v + \EE[X^{(C^{\phi_2})}] < P^* - \delta$ and $X^{(C^{\phi_2})}$ concentrates as $C$ grows large; in fact, we can show that it is so small that it overcomes the factor of $C^{1 - \phi_2}$. 

We now provide this analysis formally.

\begin{proof}[Proof of (i)]
    We bound $p_\mu(v, \CC_2)$ from below by the probability that $v$ can afford a college in $C([P^*, P^* + \delta])\subseteq C([P^*, \infty)).$
    Observe that if $\max_{c\in C(P^*, P^* + \delta)} \{v + X_c\} > P^* + \delta$, then $v$ can afford at least one college in $C([P^*, \infty))$.
    Therefore,
    \begin{equation}
        p_\mu(v, \CC_2) \ge \Pr\left[\max_{c\in C([P^*, P^* + \delta])} \{v + X_c\} > P^* + \delta\right] \ge \Pr\left[v + X^{(C^{\phi_2})} > P^* + \delta\right],
    \end{equation}
    where the last inequality follows since there are more than $C^{\phi_2}$ colleges in $C([P^*, P^* + \delta])$.

    Now observe that
    \begin{align}
        \Pr\left[v + X^{(C^{\phi_2})} > P^* + \delta\right] &= \Pr\left[X^{(C^{\phi_2})} > P^* + \delta - v\right]\\
        &\ge \Pr\left[X^{(C^{\phi_2})} > \EE\left[X^{(C^{\phi_2})}\right] - \delta \right] \label{eq:potato-1} \\ 
        &\ge \Pr\left[\left|X^{(C^{\phi_2})} - \EE\left[X^{(C^{\phi_2})}\right]\right| < \delta \right] \\
        &> 1 - \frac{\Var\left[X^{(C^{\phi_2})}\right]}{\delta^2} \label{eq:potato-2} \\
        &= 1 - O\left(\frac{C^{-\beta\phi_2}}{\delta^2}\right) \label{eq:potato-3} \\
        &= 1 - O\left(C^{-2\phi_1 - \beta\phi_2}\right),
    \end{align}
    where \eqref{eq:potato-1} follows because $v > P^* - \EE\left[X^{(C^{\phi_2})}\right]  + 2\delta,$ \eqref{eq:potato-2} follows from Chebyshev's inequality, and \eqref{eq:potato-3} follows from the fact that $\Var[X^{(n)}] = O\left(n^{-\beta}\right).$
\end{proof}

\begin{proof}[Proof of (ii).]
Note that $v$ can afford at least one college in $C(P^*, \infty)$ only if $v + X_c > P^*$ for some $c\in C(P^*, \infty).$ Thus,
    \begin{align}
        p_\mu(v, \CC_2) &\le \Pr[v + X^{(C)} > P^*]\\
        &\le C^{1 - \phi_2} \cdot \Pr[v + X^{(C^{\phi_2})} > P^*],
    \end{align}
    where the last inequality holds by a union bound (split the $C$ colleges into $C^{1-\phi_2}$ groups each of size $C^{\phi_2}$, and then bound the probability that the student is admitted to a college in any of these groups). 
    
    Then, proceeding as in part (i), 
    \begin{align}
        \Pr\left[v + X^{(C^{\phi_2})} > P^* \right] &= \Pr\left[X^{(C^{\phi_2})} > P^* - v\right]\\
        &\ge \Pr\left[X^{(C^{\phi_2})} > \EE\left[X^{(C^{\phi_2})}\right] + \delta \right] \label{eq:yam-1} \\ 
        &\ge \Pr\left[\left|X^{(C^{\phi_2})} - \EE\left[X^{(C^{\phi_2})}\right]\right| > \delta \right] \\
        &< \frac{\Var\left[X^{(C^{\phi_2})}\right]}{\delta^2} \label{eq:yam-2} \\
        &= O\left(\frac{C^{-\beta \phi_2}}{\delta^2}\right), \label{eq:yam-3}
    \end{align}
    where \eqref{eq:yam-1} follows because $v < P^* - \EE\left[X^{(C^{\phi_2})}\right]  - \delta,$ \eqref{eq:yam-2} follows from Chebyshev's inequality, and \eqref{eq:yam-3} follows from the fact that $\Var[X^{(n)}] = O\left(\frac{1}{n}\right).$ It follows that
    \begin{equation}
        p_\mu(v, \CC_2) \le C^{1 - \phi_2}\cdot O\left(\frac{C^{-\beta \phi_2}}{\delta^2}\right) = O\left(C^{1-2\phi_1-(1+\beta)\phi_2}\right).
    \end{equation}
\end{proof}

We may now use the preceding result (\Cref{prop:tomato}) to bound the second integral, the mass of students with value at most $v_S$ who can afford a college in $C([P^*,\infty))$.

\begin{proposition}\label{prop:duck-2}
\begin{equation}
    \int_{-\infty}^{v_S} p_\mu(v, \CC_2)\,d\eta(v) = O\left(C^{-K(\beta,\gamma)}\right).
\end{equation}
\end{proposition}

\begin{proof}
We split the integral into three parts:
\begin{alignat}{2}
    \int_{-\infty}^{v_S} p_\mu(v, \CC_2)\,d\eta(v) 
    &= \quad\int_{-\infty}^{P^* - \EE[X^{(C^{\phi_2})}] - \delta} && p_\mu(v, \CC_2)\,d\eta(v) \label{eq:cat-1} \\ 
    &\quad+ \int_{P^* - \EE[X^{(C^{\phi_2})}] - \delta}^{P^* - \EE[X^{(C^{\phi_2})}] + 2\delta} && p_\mu(v, \CC_2)\,d\eta(v) \label{eq:cat-2} \\
    &\quad+ \int_{P^* - \EE[X^{(C^{\phi_2})}] + 2\delta}^{v_S} && p_\mu(v, \CC_2)\,d\eta(v).\label{eq:cat-3} 
\end{alignat}
We bound these parts separately.

For \eqref{eq:cat-1}, we apply \Cref{prop:tomato}(ii) to get
\begin{align}
    \int_{-\infty}^{P^* - \EE[X^{(C^{\phi_2})}] - \delta} p_\mu(v, \CC_2)\,d\eta(v) &\le \int_{-\infty}^{P^* - \EE[X^{(C^{\phi_2})}] - \delta} O\left(C^{1-2\phi_1 - (1 + \beta)\phi_2}\right) \,d\eta(v)\\
    &\le \int_{-\infty}^\infty O\left(C^{1-2\phi_1 - (1 + \beta)\phi_2}\right) \,d\eta(v)\\
    &= O\left(C^{1-2\phi_1 - (1 + \beta)\phi_2}\right).
\end{align}

For \eqref{eq:cat-2}, we have that
\begin{align}
    \int_{P^* - \EE[X^{(C^{\phi_2})}] - 2\delta}^{P^* - \EE[X^{(C^{\phi_2})}] + \delta} p_\mu(v, \CC_2)\,d\eta(v) 
    &\le \int_{P^* - \EE[X^{(C^{\phi_2})}] - 2\delta}^{P^* - \EE[X^{(C^{\phi_2})}] + \delta} 1\,d\eta(v)\\
    &= \eta((P^* - \EE[X^{(C^{\phi_2})}] - 2\delta, P^* - \EE[X^{(C^{\phi_2})}] + \delta))\\
    &= O\left(C^{\gamma\phi_1}\right),
\end{align}
where we used that the $\eta$-measure of any interval of length $3\delta$ is $O\left(C^{\gamma\phi_1}\right)$.

For \eqref{eq:cat-3}, we proceed somewhat more indirectly. The basic idea is that because students with value above $P^* - \EE[X^{(C^{\phi_2})}] + 2\delta$ have a probability near 1 of matching, and because the mass of students with value above $v_S$ is equal to the total capacity $S$, the mass of students with value in $(P^* - \EE[X^{(C^{\phi_2})}]+2\delta, v_S)$ cannot be large since otherwise more students would match than there is capacity. With this strategy in mind, we use \Cref{prop:tomato}(i) to get
\begin{align}
    S &\ge \int_{P^* - \EE[X^{(C^{\phi_2})}] + 2\delta}^\infty p_\mu(v, \CC_2)\,d\eta(v)\\ 
    &\ge \left(1 - O\left(C^{-2\phi_1 - \beta\phi_2}\right)\right) \cdot \eta((P^* - \EE[X^{(C^{\phi_2})}] + 2\delta, \infty))\\
    &= \left(1 - O\left(C^{-2\phi_1 - \beta\phi_2}\right)\right)\cdot (S + \eta((P^* - \EE[X^{(C^{\phi_2})}] + 2\delta, v_S))).
\end{align}
Rearranging,
\begin{align}
    \eta((P^* - \EE[X^{(C^{\phi_2})}] + 2\delta, v_S)) \le \frac{S\cdot O\left(C^{-2\phi_1 - \beta\phi_2}\right)}{1 - O\left(C^{-2\phi_1 - \beta\phi_2}\right)} = O\left(C^{-2\phi_1 - \beta\phi_2}\right).
\end{align}
It follows that
\begin{equation}
    \int_{P^* - \EE[X^{(C^{\phi_2})}] + 2\delta}^{v_S} p_\mu(v, \CC_2)\,d\eta(v) \le \int_{P^* - \EE[X^{(C^{\phi_2})}] + 2\delta}^{v_S} 1\,d\eta(v) = O\left(C^{-2\phi_1 - \beta\phi_2}\right). 
\end{equation}
Using the three bounds, we have that
\begin{align}
    \int_{-\infty}^{v_S} p_\mu(v, \CC_2)\, d\eta(v) &= O\left(C^{1-2\phi_1 - (1 + \beta)\phi_2}\right) + O\left(C^{\gamma\phi_1}\right) + O\left(C^{-2\phi_1 - \beta\phi_2}\right)\\
    &= O\left(C^{1-2\phi_1 - (1 + \beta)\phi_2}\right) + O\left(C^{\gamma\phi_1}\right),
\end{align}
where we observe that $1-2\phi_1 - (1 + \beta)\phi_2 > -2\phi_1 - \beta\phi_2$ since $1 - \phi_2 = 1 - \frac{5\gamma + 6}{3\beta\gamma + 2\beta + 5\gamma + 6} > 0.$ Finally, we may verify that
\begin{align}
    1-2\phi_1 - (1 + \beta)\phi_2 &= -K(\beta,\gamma)\\
    \gamma\phi_1 &= -K(\beta,\gamma).
\end{align}
\end{proof}

Putting together \Cref{prop:duck-1} and \Cref{prop:duck-2}, we have that
\begin{align}
    \int_{-\infty}^{v_S} p_\mu(v)\,d\eta(v) &\le S(\CC_2) + \int_{-\infty}^{v_S} p(v, \CC_2)\,d\eta(v)\\
    &= O\left(C^{-K(\beta,\gamma)}\right),
\end{align}
completing the proof for Case 1.

\subsubsection*{Case 2: Many Cutoffs Below $P^*$}

Set
\begin{align}
    \CC_1 &:= C((-\infty, P_{C^{\phi_3}}))\\
    \CC_2 &:= C((P_{C^{\phi_3}}, \infty)).
\end{align}
As in the first case,
\begin{equation}
    \int_{v_S}^\infty p_\mu(v)\,d\eta(v) \le S(\CC_1) + \int_{v_S}^\infty p_\mu(v,\CC_2)\,d\eta(v),
\end{equation}

By the pigeonhole principle,
\begin{equation}
    |P_{C^{\phi_3}} - P_1| > \delta \cdot \frac{C^{\phi_3}}{C^{\phi_2}} = C^{\phi_1 + \phi_3 - \phi_2}.
\end{equation}
This gives a result that will be key for later in the proof: that the cutoff of college $C^{\phi_3}$ is significantly larger than $P_1.$

\begin{proposition}\label{prop:goose-1}
\begin{equation}
    S(\CC_1) < \alpha C^{-\frac{2\beta\gamma}{3\beta\gamma + 2\beta + 5\gamma + 6}}.
\end{equation}
\end{proposition}
\begin{proof}
    The capacity of each college is at most $\frac{\alpha}{C}$. Since there are fewer than $C^{\phi_3}$ colleges in $C((-\infty, P_{C^{\phi_3}}))$, the mass of students who match to these colleges is less than $C^{\phi_3}\cdot \frac{\alpha}{C} = \alpha C^{1-\phi_3}.$ The result follows since
    \begin{equation}
        1 - \phi_3 = -\frac{2\beta\gamma}{3\beta\gamma + 2\beta + 5\gamma + 6}.
    \end{equation}
\end{proof}

\begin{lemma}\label{lem:log-bound}
    If $\DD$ is $\beta$-max-concentrating for some $\beta>0$, then
    \begin{equation}
        \EE[X^{(n)}] = o(\log n).
    \end{equation}
\end{lemma}

\begin{proof}
    It suffices to show that for all $\epsilon > 0,$
    \begin{equation}
        \EE[X^{(2n)}] < \EE[X^{(n)}] + \epsilon
    \end{equation}
    for all $n$ sufficiently large.\footnote{To see this, note that for any $\epsilon>0,$ $f(2n) < f(n) + \epsilon$ for all $n$ sufficiently large implies that for $g(x):=f(e^x)$, $g(\log n + \log 2) < g(\log n) + \epsilon$ for all $n$ sufficiently large. This shows that $g(x)=o(x).$ Since $g(x)=f(e^x)$, then $g(\log x)=f(x)$, so $f(x)=o(\log x)$.}
    We have that
    \begin{align}
        \EE[X^{(2n)}] = \EE[\max\{W_1, W_2\}],
    \end{align}
    where $W_1$ and $W_2$ are drawn i.i.d. from the same distribution as $X^{(n)}.$ Then,
    \begin{align}
        \EE[\max\{W_1,W_2\}]
        &\le \EE[W_1 + |W_2 - W_1|]\\
        &= \EE[W_1] + \EE[|W_2 - W_1|]\\
        &= \EE[X^{(n)}] + \EE[|W_2 - \EE[X^{(n)}]| + |W_1 - \EE[X^{(n)}]|]\\
        &= \EE[X^{(n)}] + 2\EE[|X^{(n)} - \EE[X^{(n)}]|]\\
        &\le \EE[X^{(n)}] + 2\sqrt{\Var[X^{(n)}]}.
    \end{align}
    The result follows since $\lim_{n\rightarrow \infty} \Var[X^{(n)}] = 0$.
\end{proof}

We now introduce another constant,
\begin{equation}
    \phi_4 = -\frac{\beta}{2} + \frac{\beta\gamma}{3\beta\gamma + 2\beta + 5\gamma + 6}.
\end{equation}
Note that $\phi_4 < \frac{\beta\gamma}{3\beta\gamma + 2\beta + 5\gamma + 6} = \phi_1 + \phi_3 - \phi_2.$
\begin{proposition}\label{prop:beet}
\quad
\begin{itemize}
    \item[(i)] If $v < P_{C^{\phi_3}} - \EE[X^{(C)}] - C^{\phi_4},$
    \begin{equation}
        p_\mu(v, \CC_2) \le O\left(C^{-\beta-2\phi_4}\right).
    \end{equation}
    \item[(ii)] If $v > P_{C^{\phi_3}} - \EE[X^{(C)}] - C^{\phi_4},$
    \begin{equation}
        p_\mu(v) \ge 1 - O\left(C^{-2\phi_1 - 2\phi_3 + 2\phi_2}\right).
    \end{equation}
\end{itemize}
\end{proposition}

\begin{proof}[Proof of (i).]
    The probability that $v$ can afford some college in $C([P_{C^{\phi_3}}, \infty))$ is bounded from above by the probability that $v + X_c > P_{C^{\phi_3}}$ for some $c\in C$. Therefore,
    \begin{align}
        p_\mu(v, \CC_2) &\le \Pr[v + X^{(C)} > P_{C^{\phi_3}}]\\
        &= \Pr[X^{(C)} > P_{C^{\phi_3}} - v]\\
        &\le \Pr\left[X^{(C)} > \EE[X^{(C)}] + C^{\phi_4}\right] \label{eq:penguin-1} \\
        &\le \Pr\left[|X^{(C)} - \EE[X^{(C)}]| > C^{\phi_4} \right]\\
        &\le \frac{\Var[X^{(C)}]}{C^{2\phi_4}} \label{eq:penguin-2}\\
        &\le O\left(C^{-\beta - 2\phi_4}\right), \label{eq:penguin-3}
    \end{align}
    where \eqref{eq:penguin-1} follows because $v < P_{C^{\phi_3}} - \EE[X^{(C)}] - C^{\phi_4}$, \eqref{eq:penguin-2} follows from Chebyshev's inequality, and \eqref{eq:penguin-3} follows from recalling that $\Var[X^{(n)}] = O\left(n^{-\beta}\right).$
\end{proof}

\begin{proof}[Proof of (ii).]
    The probability $v$ is matched is bounded from below by the probability that they match with the college with cutoff $P_1$. Therefore,
    \begin{align}
        p_\mu(v) &\ge \Pr[v + X > P_1]\\
        &= \Pr[X > P_1 - v]\\
        &\ge \Pr\left[X > P_1 - P_{C^{\phi_3}} + \EE[X^{(C)}] + C^{\phi_4}\right] \label{eq:poppin-1} \\
        &\ge \Pr\left[X > - C^{\phi_1 + \phi_3 - \phi_2} + O\left(\log C \right) + C^{\phi_4}\right] \label{eq:poppin-2} \\
        &= \Pr[X > - C^{\phi_1 + \phi_3 - \phi_2} + o(C^{\phi_1 + \phi_3 - \phi_2})]\\
        &= 1 - \Pr[X < - \Theta(C^{\phi_1 + \phi_3 - \phi_2})]
    \end{align}
    where \eqref{eq:poppin-1} follows because $v > P_{C^{\phi_3}} - \EE[X^{(C)}] - C^{\phi_4}$ and \eqref{eq:poppin-2} follows from recalling that $P_1 - P_{C^{\phi_3}} > C^{\phi_1 + \phi_3 - \phi_2}$ and $\EE[X^{(C)}] = o\left(\log C\right)$ (\Cref{lem:log-bound}). Now observe that $\Pr[X < - \Theta(C^{\phi_1 + \phi_3 - \phi_2})] = \Pr[X - \EE[X] < - \Theta(C^{\phi_1 + \phi_3 - \phi_2})],$ since $\EE[X]$ is fixed in $C$. Therefore,
    \begin{align}
        1 - \Pr[X < - O(C^{\phi_1 + \phi_3 - \phi_2})] &= 1 - \Pr[X - \EE[X] < - \Theta(C^{\phi_1 + \phi_3 - \phi_2})]\\
        &\ge 1 - \Pr[|X - \EE[X]| > \Theta(C^{\phi_1 + \phi_3 - \phi_2})]\\
        &\ge 1 - \frac{\Var[X]}{\Theta(C^{\phi_1 + \phi_3 - \phi_2})} \label{eq:poppin-3} \\
        &= 1 - O\left(C^{-2\phi_1 -2\phi_3 + 2\phi_2}\right), \label{eq:poppin-4}
    \end{align}
    where \eqref{eq:poppin-3} follows from Chebyshev's inequality, and \eqref{eq:poppin-4} follows because $\Var[X]$ is a constant.
\end{proof}

\begin{proposition}\label{prop:goose-2}
    \begin{equation}
    \int_{-\infty}^{v_S} p_\mu(v, \CC_2)\,d\eta(v) = O\left(C^{-K(\beta,\gamma)}\right).
\end{equation}
\end{proposition}

\begin{proof}
    We split the integral into three parts:
    \begin{alignat}{2}
        \int_{-\infty}^{v_S} p_\mu(v, \CC_2)\,d\eta(v) 
        &= \quad\int_{-\infty}^{P_{C^{\phi_3}} - \EE[X^{(C)}] - C^{\phi_4}} && p_\mu(v, \CC_2)\,d\eta(v) \label{eq:llama-1} \\ 
        &\quad+ \int_{P_{C^{\phi_3}} - \EE[X^{(C)}] - C^{\phi_4}}^{v_S} && p_\mu(v, \CC_2)\,d\eta(v).\label{eq:llama-2} 
    \end{alignat}
    We bound the two parts separately.

    For \eqref{eq:llama-1}, we apply \Cref{prop:beet}(i) to get
    \begin{align}
        \int_{-\infty}^{P_{C^{\phi_3}} - \EE[X^{(C)}] - C^{\phi_4}} p_\mu(v, \CC_2)\,d\eta(v) &\le \int_{-\infty}^{P_{C^{\phi_3}} - \EE[X^{(C)}] - C^{\phi_4}} O\left(C^{-\beta-2\phi_4}\right)\,d\eta(v)\\
        &\le \int_{-\infty}^\infty O\left(C^{-\beta-2\phi_4}\right)\,d\eta(v)\\
        &= O\left(C^{-\beta-2\phi_4}\right) = O\left(C^{-K(\beta,\gamma)}\right),
    \end{align}
    where we note that $\beta+2\phi_4 = K(\beta,\gamma).$

    For \eqref{eq:llama-2}, we proceed with the same idea used to bound \eqref{eq:cat-3} in \Cref{prop:duck-2}. Again, the basic idea is that because students with value above $P_{C^{\phi_3}} - \EE[X^{(C)}] - C^{\phi_4}$ have a probability near 1 of matching, and because the measure of students with value above $v_S$ is equal to the total capacity, the measure of students with value in $(P_{C^{\phi_3}} - \EE[X^{(C)}] - C^{\phi_4}, v_S)$ cannot be large since otherwise more students would match than there is capacity. With this strategy in mind, we use \Cref{prop:beet}(ii) to get
    \begin{align}
        S &\ge \int_{P_{C^{\phi_3}} - \EE[X^{(C)}] - C^{\phi_4}}^\infty p_\mu(v)\,d\eta(v)\\ 
        &\ge \left(1 - O\left(C^{-2\phi_1 -2\phi_3 + 2\phi_2}\right)\right) \cdot \eta\left(\left(P_{C^{\phi_3}} - \EE[X^{(C)}] - C^{\phi_4}, \infty\right)\right)\\
        &= \left(1 - O\left(C^{-2\phi_1 -2\phi_3 + 2\phi_2}\right)\right)\cdot \left(S + \eta\left(\left(P_{C^{\phi_3}} - \EE[X^{(C)}] - C^{\phi_4}, v_S\right)\right)\right).
    \end{align}
    Rearranging,
    \begin{align}
        \eta((P_{C^{\phi_3}} - \EE[X^{(C)}] - C^{\phi_4}, v_S)) \le \frac{S\cdot O\left(C^{-2\phi_1 -2\phi_3 + 2\phi_2}\right)}{1 - O\left(C^{-2\phi_1 -2\phi_3 + 2\phi_2}\right)} = O\left(C^{-2\phi_1 -2\phi_3 + 2\phi_2}\right).
    \end{align}
    It follows that
    \begin{align}
        \int_{P_{C^{\phi_3}} - \EE[X^{(C)}] - C^{\phi_4}}^{v_S} p_\mu(v, \CC_2)\,d\eta(v) &\le \int_{P_{C^{\phi_3}} - \EE[X^{(C)}] - C^{\phi_4}}^{v_S} 1\,d\eta(v)\\
        &= O\left(C^{-2\phi_1 -2\phi_3 + 2\phi_2}\right) = O\left(C^{-K(\beta,\gamma)}\right),
    \end{align}
    where we note that $-2\phi_1 -2\phi_3 + 2\phi_2 = -K(\beta,\gamma).$
    
    The result follows, as we have shown that both \eqref{eq:llama-1} and \eqref{eq:llama-2} are $O(C^{-K(\beta,\gamma)})$.
\end{proof}

Putting together \Cref{prop:goose-1} and \Cref{prop:goose-2}, we have that
\begin{align}
    \int_{-\infty}^{v_S} \Pr[\mu(v)\in C]\,d\eta(v) &\le S(\CC_1) + \int_{-\infty}^{v_S} p_\mu(v, \CC_2)\,d\eta(v)\\
    &= O\left(C^{-K(\beta,\gamma)}\right),
\end{align}
completing the proof for Case 2.

%% file: proof-amplifying.tex
\section{Proof of Theorem \ref{thm:amplifying} (Full Amplification)}\label{sec:amplifying}

First observe that \Cref{thm:amplifying} is equivalent to the following result.

\begin{proposition}\label{thm2v2}
Let $\DD$ be long-tailed. Then, for all $C$ sufficiently large, for all $\epsilon>0$ and any $\mu\in M(\DD, C)$,
\begin{equation}
    \Pr[\mu(v)\in C]\in (S-\epsilon, S+\epsilon).
\end{equation}
for all $v\in (v_\epsilon^-, v_\epsilon^+)$ for $v_\epsilon^-, v_\epsilon^+$ chosen such that $\eta(-\infty, v_\epsilon^-) = \eta((v_\epsilon, \infty)) = \epsilon$.
\end{proposition}

To see why this is equivalent, observe that for any $v$, when $\epsilon$ is taken sufficiently small, $v\in (v_\epsilon^-, v_\epsilon^+).$ Then for $C$ sufficiently large, $\Pr[\mu(v)\in C]\in (S-\epsilon, S+\epsilon).$

We begin with some high-level intuition for why long-tailed distributions induce this behavior. Consider two students with values $v_1$ and $v_2$. Suppose that $v_2 = v_1 - d$ for some $d > 0$. Now consider the likelihood that each applicant can afford a firm with cutoff $P$. These probabilities are $\Pr[v_1 + X > P]$ and $\Pr[v_2 + X > P]$ respectively. Consider the ratio
\begin{equation}
    \frac{\Pr[v_2 + X > P]}{\Pr[v_1 + X > P]} = \frac{\Pr[X > P - v_1 + d]}{\Pr[X > P - v_1]}.
\end{equation}
Then since $\DD$ is long-tailed,
\begin{equation}
    \lim_{P\rightarrow \infty} \frac{\Pr[X > P - v_1 + d]}{\Pr[X > P - v_1]} = 1.
\end{equation}
This tells us that if a college's cutoff is large, two students (with different values) can afford the college with approximately equal probability. To leverage this fact, we will show that almost all colleges must have large cutoffs. Roughly speaking, this is true because if there were a large number of colleges with small cutoffs, then more students would be able to afford these firms than there is total capacity.

\subsection*{Proof of \Cref{thm2v2}}
We consider a long-tailed probability distribution $\DD$. To prove \Cref{thm2v2}, we show the following equivalent statement: There exists a constant $\sigma > 0$ such that for all $C$ sufficiently large, if $\mu \in M(\DD, C),$ 
\begin{equation}\label{eq:equiv-lt}
    p_\mu(v) \in \left(S - (1 + \alpha)\epsilon - (1 - e^{-2\epsilon\sigma}), S + \epsilon + (1 - e^{-2\epsilon\sigma}) + \frac{\alpha\sqrt{\epsilon}}{1 - S - \epsilon - (1 - e^{-2\epsilon\sigma})}\right)
\end{equation}
for all $v\in (v_-,v^*)$, where $v_-, v_+$ are chosen such that $\eta(-\infty,v_-) = \eta(v_+,\infty) = \frac{\epsilon}{2}$ and $v^*$ is chosen such that $\eta(v^*,v_+)=\sqrt{\epsilon}.$ Indeed, the result follows since the interval in \eqref{eq:equiv-lt} is a vanishingly small interval around $S$ as $\epsilon\rightarrow 0$, and since every $v$ is contained in $(v_-, v^*)$ for $\epsilon$ sufficiently small.
\\
\\
Consider $\mu \in M(\DD, C)$, with cutoffs $P_1\le P_2 \le \cdots \le P_C.$ Define the following two sets of colleges: 
\begin{align}
    \FF_1 &:= \{1, 2, \cdots, \epsilon C\}\\
    \FF_2 &:= \{\epsilon C + 1, \epsilon C + 2, \cdots, C\}.
\end{align}
Then observe that
\begin{equation}\label{eq:lt-bounds}
    p_\mu(v, \CC_2) \le p_\mu(v) \le p_\mu(v, \CC_1) + p_\mu(v, \CC_2). 
\end{equation}
We use \eqref{eq:lt-bounds} to show the desired result. To briefly outline our plan: First, we show that $\Pr[\mu(v)\in \FF_1]$ must be small for all $v$ except on a set of small $\eta$-measure. Since $\FF_1$ is a small fraction of the firms in the market, $p_\mu(v, \CC_1)$ can be thought of as an ``error term.'' Second, and more substantially, we show that $p_\mu(v, \CC_2) \approx S$ for all $v$ except on a set of small $\eta$-measure. Then the result will follow by applying the bounds in \eqref{eq:lt-bounds}.
\\
\\
We focus first on $p_\mu(v, \CC_2)$. We will show the following result.
\begin{proposition}\label{prop:lt-large-firms}
There exists a constant $\sigma > 0$ such that for $C$ sufficiently large,
    \begin{equation}\label{eq:lt-large-firms}
    p_\mu(v, \CC_2) \in (S - (1 + \alpha)\epsilon - (1 - e^{-2\epsilon\sigma}), S + \epsilon + (1 - e^{-2\epsilon\sigma})).
    \end{equation}
for all $v\in (v_-, v_+).$
\end{proposition}

Observe that by using the bounds in \Cref{prop:lt-small-firms} and \Cref{prop:lt-large-firms} in combination with \eqref{eq:lt-bounds}, we obtain the desired result \eqref{eq:equiv-lt}. Therefore, it suffices to show \Cref{prop:lt-large-firms}, which we devote the remainder of this section towards.

\paragraph{Proof of \Cref{prop:lt-large-firms}.} 
We first show a few useful lemmas. The first lemma demonstrates that as $C$ is taken large, almost all college cutoffs are must also be large.

\begin{lemma}\label{lem:unbounded-cutoffs}
    For any $P\in \RR$, $P_{\epsilon C} > P$ for all $C$ sufficiently large.
\end{lemma}

\begin{proof}
    Suppose otherwise. Then there exists $P$ such that there exists $C$ arbitrarily large and $\mu \in M(\DD, C)$ with cutoffs $P_1\le P_2\le \cdots \le P_C$ such that $P_{\epsilon C} \le P$.
    Then we show that the number of applicants who can afford a firm in $\FF_1$ exceeds the total capacity $S$. Indeed,
    \begin{align}
        \int_{-\infty}^\infty p_\mu(v, \CC_1)\,d\eta(v)
        &= \int_{-\infty}^\infty \Pr[v + X_c > P_c\text{ for some $c\in \FF_1$}]\,d\eta(v)\\
        &\ge \int_{-\infty}^\infty \Pr[v + X^{(\epsilon C)} > P_{\epsilon C}]\,d\eta(v)\\
        &\ge \int_{-\infty}^\infty \Pr[v + X^{(\epsilon C)} > P]\,d\eta(v).
    \end{align}
    Since $\DD$ is unbounded from above, the final integral approaches $1$ as $C$ grows large. So for $C$ sufficiently large,
    \begin{equation}
        \int_{-\infty}^\infty p_\mu(v, \CC_1)\,d\eta(v) > S,
    \end{equation}
    which contradicts the fact that the total capacity of colleges is $S$.
\end{proof}

We now use \Cref{lem:unbounded-cutoffs} to show that the probability that applicants with values $v_-$ and $v_+$ have approximately equal chances of affording a college in $\FF_2$ when $C$ is sufficiently large.

\begin{lemma}\label{lem:apple}
    For $C$ sufficiently large,
    \begin{equation}
        \frac{\Pr[v_- + X > P_c]}{\Pr[v_+ + X > P_c]} > 1 - \epsilon
    \end{equation}
    for all $c\in \FF_2.$
\end{lemma}

\begin{proof}
    For $P$ sufficiently large,
    \begin{equation}
        \frac{\Pr[v_- + X > P]}{\Pr[v_+ + X > P]} > 1 - \epsilon
    \end{equation}
    since $\DD$ is long-tailed. By \Cref{lem:unbounded-cutoffs}, for any $P\in \RR$, $P_{\epsilon C} > P$ for all $C$ sufficiently large. Observing that $P_c\ge P_{\epsilon C}$ for all $c\in \FF_2,$ the result follows.
\end{proof}

\Cref{lem:apple} shows that the probability of affording a firm in $\FF_2$ is approximately the same for $v_-$ and $v_+$ We require a slightly stronger result: that the probability of affording \textit{at least one} college in $\FF_2$ is approximately the same. To show this, we also require the following bound.

\begin{lemma}
    There exists a constant $\sigma>0$ such that for $C$ sufficiently large,
    \begin{equation}
        \Pr[v_+ + X > P_c] \le \frac{\sigma}{C}
    \end{equation}
    for all $c\in \FF_2.$
\end{lemma}

\begin{proof}
    Assume otherwise, such that such that for all $\sigma > 0$, there exists $C$ arbitrarily large and $\mu \in M(\DD, C)$ with cutoffs $P_1\le P_2\le \cdots \le P_C$ such that
    \begin{equation}
        \Pr[v_+ + X > P_{\epsilon C + 1}] > \frac{\sigma}{C}.
    \end{equation}

    Then if we take $C$ and $\sigma$ sufficiently large, and choose such a $\mu,$ we have that
    \begin{align}
        \int_{-\infty}^\infty p_\mu(v, \CC_1)\,d\eta(v)
        &= \int_{v_-}^{v_+} \Pr[v + X_c > P_c\text{ for some $c\in \FF_1$}]\,d\eta(v)\\
        &\ge (1 - \epsilon) \Pr[v_- + X_c > P_{\epsilon c}\text{ for some $c\in \FF_1$}] \label{eq:building-1}\\
        &= (1-\epsilon)\left( 1 - \left(1 - \Pr[v_- + X > P_{\epsilon m}]\right)^{|\FF_1|}\right) \label{eq:building-2}\\
        &> (1-\epsilon)\left(1 - \left(1 - (1-\epsilon)\Pr[v_+ + X > P_{\epsilon C}]\right)^{|\FF_1|}\right) \label{eq:building-3}
    \end{align}
    where \eqref{eq:building-1} follows by noting that $\eta(v_-, v_+)=1-\epsilon$ and that $\Pr[v + X_c > P_c\text{ for some $c\in \FF_1$}] \ge \Pr[v_- + X_c > P_c\text{ for some $c\in \FF_1$}]$ for $v\in (v_-,v_+).$ \eqref{eq:building-3} follows from taking $C$ sufficiently large and applying \Cref{lem:apple}. Continuing,
    \begin{align}
        (1-\epsilon)\left(1 - \left(1 - (1-\epsilon)\Pr[v_+ + X > P_{\epsilon C}]\right)^{|\FF_1|}\right)
        &> (1-\epsilon)\left(1 - \left(1 - \frac{(1-\epsilon)\sigma}{C}\right)^{|\FF_1|}\right)\\
        &> (1-\epsilon)\left(1 - \exp\left[- \frac{|\FF_1|(1-\epsilon)\sigma}{C}\right]\right)\\
        &= (1-\epsilon)\left(1 - \exp\left[-\epsilon(1-\epsilon)\sigma\right]\right)\\
        &> S
    \end{align}
    for $\sigma$ sufficiently large when $\epsilon < 1 - S$. This provides the desired contradiction, since the number of applicants who can afford a college in $\FF_2$ cannot exceed the total capacity of colleges.
\end{proof}

We may now show that the probability $v_-$ and $v_+$ can afford at least one college in $\FF_2$ is approximately the same when $C$ is large.
\begin{proposition}\label{prop:lt-approx-F2}
For $C$ sufficiently large,
\begin{equation}
    p_\mu(v_+, \CC_2) - p_\mu(v_-, \CC_2) < 1 - e^{-2\epsilon\sigma}.
\end{equation}
\end{proposition}
\begin{proof}
We have that
\begin{align}
    p_\mu(v_+, \CC_2) - p_\mu(v_-, \CC_2)
    &= \prod_{c\in \FF_2}\Pr[v_- + X < P_c] - \prod_{c\in \FF_2}\Pr[v_+ + X < P_c]\\
    &= \left(1 - \prod_{c\in \FF_2}\frac{\Pr[v_+ + X < P_c]}{\Pr[v_- + X < P_c]}\right) \prod_{c\in \FF_2}\Pr[v_- + X < P_c]\\
    &< 1 - \prod_{c\in \FF_2}\frac{\Pr[v_+ + X < P_c]}{\Pr[v_- + X < P_c]}
\end{align}
From \Cref{lem:apple},
\begin{equation}
    \Pr[v_- + X > P_c] > (1-\epsilon)\Pr[v_+ + X > P_c]
\end{equation}
for all $c\in \FF_2$. Therefore,
\begin{align}
    \frac{\Pr[v_+ + X < P_c]}{\Pr[v_- + X < P_c]} &= \frac{1 - \Pr[v_+ + X > P_c]}{1 - \Pr[v_- + X > P_c]}\\
    &> \frac{1 - \Pr[v_+ + X > P_c]}{1 - (1-\epsilon)\Pr[v_+ + X > P_c]}.
\end{align}
Continuing,
\begin{align}
    \frac{1 - \Pr[v_+ + X > P_c]}{1 - (1-\epsilon)\Pr[v_+ + X > P_c]} &= 1 - \frac{\epsilon \Pr[v_+ + X > P_c]}{1 - (1-\epsilon) \Pr[v_+ + X > P_c]}\\
    &\ge 1 - \frac{\epsilon\frac{\sigma}{C}}{1 - (1-\epsilon)\frac{\sigma}{C}}\\
    &= 1 - \frac{\epsilon\sigma}{C - (1-\epsilon)\sigma}\\
    &> \exp\left[-\frac{2\epsilon\sigma}{C}\right]
\end{align}
where the last inequality holds when $C$ is sufficiently large. Then we have that
\begin{align}
    1 - \prod_{c\in \FF_2}\frac{\Pr[v_+ + X < P_c]}{\Pr[v_- + X < P_c]}
    &< 1 - \left(\exp\left[-\frac{2\epsilon\sigma}{C}\right]\right)^{|\FF_2|}\\
    &< 1 - e^{-2\epsilon\sigma}.
\end{align}
\end{proof}

We now use \Cref{prop:lt-approx-F2} to show \Cref{prop:lt-large-firms}, where primary remaining observation is that the average probability a student with value in $(v_-,v_+)$ can afford a college in $\FF_2$ must be approximately $S.$ Then, since the probabilities for $v$ in this interval must be about the same (by \Cref{prop:lt-approx-F2}), all the probabilities in this interval must be about $S$.

Formally, we observe that
\begin{align}
    (1 - \epsilon) p_\mu(v_-, \CC_2) \le \int_{v_-}^{v_+} p_\mu(v, \CC_2)\,d\eta(v) \le S,
\end{align}
so
\begin{equation}\label{eq:bbq-1}
    p_\mu(v_-, \CC_2) \le \frac{S}{1-\epsilon} < S + \epsilon,
\end{equation}
where the last inequality holds for $\epsilon$ sufficiently small. We also have that
\begin{align}
    p_\mu(v_+, \CC_2) &\ge \int_{v_-}^{v_+} \Pr[B(v)\cap \FF_2\neq \emptyset]\,d\eta(v)\\
    &\ge \int_{-\infty}^\infty \Pr[B(v)\cap \FF_2\neq \emptyset]\,d\eta(v) - \epsilon \\
    &\ge \int_{-\infty}^\infty \Pr[\mu(v)\cap \FF_2\neq \emptyset]\,d\eta(v) - \epsilon \label{eq:ginger}\\
    &= \left(S - \int_{-\infty}^\infty \Pr[\mu(v)\in \FF_1\neq \emptyset]\,d\eta(v)\right) - \epsilon\\
    &\ge S - \alpha\epsilon - \epsilon,
\end{align}
so
\begin{equation}\label{eq:bbq-2}
    p_\mu(v_+, \CC_2) \ge S - \alpha\epsilon - \epsilon.
\end{equation}
Then, since
\begin{equation}
    p_\mu(v_+, \CC_2) - p_\mu(v_-, \CC_2) < 1 - e^{-2\epsilon\sigma}
\end{equation}
(by \Cref{prop:lt-approx-F2}), we have that
\begin{align}
    p_\mu(v_+, \CC_2) &< p_\mu(v_-, \CC_2) + (1 - e^{-2\epsilon\sigma})\\
    &< S + \epsilon + (1 - e^{-2\epsilon\sigma})\label{eq:cucumber-1},
\end{align}
where the last line follows from \eqref{eq:bbq-1}, and that
\begin{align}
    \Pr[B(v_-)\cap \FF_2\neq \emptyset] &> \Pr[B(v_+)\cap \FF_2\neq \emptyset] - (1 - e^{-2\epsilon\sigma})\\
    &\ge S - \alpha\epsilon - \epsilon - (1 - e^{-2\epsilon\sigma}),\label{eq:cucumber-2}
\end{align}
where the last line follows from \eqref{eq:bbq-2}.
Finally, since 
\begin{equation}
    p_\mu(v_-, \CC_2) < p_\mu(v, \CC_2) < p_\mu(v_+, \CC_2),
\end{equation}
it follows from \eqref{eq:cucumber-1} and \eqref{eq:cucumber-2} that
\begin{equation}
    p_\mu(v, \CC_2) \in (S - \alpha\epsilon - \epsilon - (1 - e^{-2\epsilon\sigma}), S + \epsilon + (1 - e^{-2\epsilon\sigma})),
\end{equation}
as desired.
\\
\\
We now analyze $p_\mu(v, \CC_1)$.
\begin{proposition}\label{prop:lt-small-firms}
\begin{equation}
    p_\mu(v, \CC_1) \le \frac{\alpha\sqrt{\epsilon}}{1 - S - \epsilon - (1 - e^{-2\epsilon\sigma})}
\end{equation}
for all $v < v^*$ where $\eta((v^*, v_+))=\sqrt{\epsilon}.$
\end{proposition}
\begin{proof}
Note that any student who cannot afford any college in $\CC_2$ but can afford a college in $\CC_1$ will be matched to a college in $\CC_1.$ Therefore, the measure of students with value in $(v_*, v_+)$ who match with a college in $\CC_1$ is at least
\begin{equation}\label{eq:monkey}
    \int_{v^*}^{v_+} (1 - p_\mu(v, \CC_2)) \cdot p_\mu(v, \CC_1)\,d\eta(v).
\end{equation}
We have from \eqref{eq:cucumber-1} that
\begin{equation}
    1 - p_\mu(v, \CC_2) \ge 1 - S - \epsilon - (1 - e^{-2\epsilon\sigma})
\end{equation}
for all $v\in (v^*, v_+)\subseteq (v_-, v_+).$ 

Also $S(\CC_1)\le \epsilon C\cdot \frac{\alpha}{C} = \epsilon\alpha.$ \eqref{eq:monkey} can be at most $S(\CC_1),$ so we have that
\begin{align}
    \epsilon\alpha &\ge \int_{v^*}^{v_+} (1 - p_\mu(v, \CC_2)) \cdot p_\mu(v, \CC_1)\,d\eta(v)\\
    &\ge \int_{v^*}^{v_+} \left(1 - S - \epsilon - (1 - e^{-2\epsilon\sigma})\right) \cdot p_\mu(v, \CC_1)\,d\eta(v)\\
    &\ge \eta((v^*,v_+)) \left(1 - S - \epsilon - (1 - e^{-2\epsilon\sigma})\right) \cdot p_\mu(v^*, \CC_1)\\
    &= \sqrt{\epsilon} \left(1 - S - \epsilon - (1 - e^{-2\epsilon\sigma})\right) \cdot p_\mu(v^*, \CC_1).
\end{align}
\end{proof}
It follows that
\begin{equation}
    p_\mu(v^*, \CC_1) \le \frac{\alpha\sqrt{\epsilon}}{1 - S - \epsilon - (1 - e^{-2\epsilon\sigma})}.
\end{equation}
The result follows since $p_\mu(v^*, \CC_1)$ is monotonically increasing in $v$.

%% file: proofs-extended.tex
\section{Proofs of Extended Results}\label{sec:proofs-extended}

We now turn to the proofs of \Cref{thm:attenuating-extended} and \Cref{thm:amplifying-extended}. Both proofs follow almost immediately from results we have already shown in the proofs of \Cref{thm:attenuating} and \Cref{thm:amplifying}.

\subsection{Proof of Theorem \ref{thm:attenuating-extended}}
Let $\DD$ be max-concentrating. We would like to show that for any $\epsilon > 0$, there exists $C(\epsilon)$ such that the following holds. Let $C$ be a $(\DD, \eta)$-coalition in an economy $E$ with stable matching $\mu$ such that $|C|>C(\epsilon).$ Then there exists $v'\in \RR$ and $C'\subseteq C$ with $\frac{|C'|}{|C|} > 1 - \epsilon$ such that
\begin{equation}\label{eq:penguin2-1}
    p_\mu(v, C') < \epsilon \qquad \text{for all $v < v'-\epsilon$}
\end{equation}
and
\begin{equation}\label{eq:penguin2-2}
    p_\mu(v, C') > 1 - \epsilon \qquad \text{for all $v > v'+\epsilon$}.
\end{equation}

Suppose, without loss of generality, that $C = \{1, 2, \cdots, C\}$, and as in the proof of \Cref{thm:attenuating}, let $P_1\le P_2\le \cdots \le P_C$ be the market-clearing cutoffs corresponding to $\mu$ at the colleges in the coalition $C$. Then we proceed in the two cases as in the proof of \Cref{thm:attenuating} (specifically, \Cref{thm1v2}). 

In the first case, the result follows directly from \Cref{prop:tomato}, where we set $C'=\CC_2$ and $v' = P^* - \EE[X^{C^{\phi_2}}]$. Indeed, \eqref{eq:penguin2-1} follows directly from \Cref{prop:tomato}(ii) and \eqref{eq:penguin2-2} follows directly from \Cref{prop:tomato}(i), where we note in the latter case that $p_\mu(v, C) > p_\mu(v, \CC_2).$ In the second case, the result follows directly from \Cref{prop:beet}, where we set $C'=\CC_2$ and $v' = P_{C^{\phi_3}} - \EE[X^{(C)}] - C^{\phi_4}$.

\subsection{Proof of Theorem \ref{thm:amplifying-extended}}
Let $\DD$ be long-tailed. We would like to show that for any $\epsilon > 0$, there exists $C(\epsilon)$ such that the following holds. Let $C$ be a $(\DD, \eta)-$coalition in an economy $E$ with stable matching $\mu$ such that $|C|>C(\epsilon).$ Then there exists $S'\in \RR$ and $C'\subseteq C$ with $\frac{|C'|}{|C|} > 1 - \epsilon$, such that
\begin{equation}\label{eq:poppin}
    |p_\mu(v, C') - S'| < \epsilon
\end{equation}
for all $v$ except on a set of $\eta$-measure at most $\epsilon.$ Equivalently, there exists a set $V$ such that $\eta(V) > 1- \epsilon$ such that \eqref{eq:poppin} holds for all $v\in V$.

Suppose, without loss of generality, that $C = \{1, 2, \cdots, C\}$, and as in the proof of \Cref{thm:attenuating}, let $P_1\le P_2\le \cdots \le P_C$ be the market-clearing cutoffs corresponding to $\mu$ at the colleges in the coalition $C$. We then proceed as in the proof of \Cref{thm:amplifying} (specifically, \Cref{thm2v2}). Taking $V = (v_-, v_+)$ and $C'=\CC_2$, the result follows directly from \Cref{prop:lt-approx-F2}, which shows that for $C$ sufficiently large,
\begin{equation}
    p_\mu(v_+, \CC_2) - p_\mu(v_-, \CC_2) < 1 - e^{-2\epsilon \sigma}.
\end{equation}
Then since $p_\mu(v, \CC_2)$ is monotonically increasing on the interval $[v_-, v_+],$ and since $\eta([v_-, v_+])$ can be taken arbitrarily close to $1$, the result follows.